\documentclass[acmsmall,screen,authorversion,nonacm]{acmart}
\setcopyright{rightsretained}
\acmPrice{}
\acmDOI{10.1145/3498696}
\acmYear{2022}
\copyrightyear{2022}
\acmSubmissionID{popl22main-p244-p}
\acmJournal{PACMPL}
\acmVolume{6}
\acmNumber{POPL}
\acmArticle{35}
\acmMonth{1}


\author{Yihong Zhang}
\affiliation{
  \institution{University of Washington}
 \country{USA}
}
\email{yz489@cs.washington.edu}

\author{Yisu Remy Wang}
\affiliation{
  \institution{University of Washington}
 \country{USA}
}
\email{remywang@cs.washington.edu}

\author{Max Willsey}
\orcid{0000-0001-8066-4218}
\affiliation{
  \institution{University of Washington}
 \country{USA}
}
\email{mwillsey@cs.washington.edu}

\author{Zachary Tatlock}
\affiliation{
  \institution{University of Washington}
  \country{USA}
}
\email{ztatlock@cs.washington.edu}

\usepackage[inline]{enumitem}
\usepackage{caption}
\usepackage{subcaption}
\usepackage{tikzit}
\tikzstyle{node}=[fill=white, draw=none, shape=circle]

\usepackage{graphicx}
\graphicspath{ {./figures/} }

\usepackage{syntax}

\usepackage{pifont}

\usepackage[colorinlistoftodos]{todonotes}
\usepackage[ruled,vlined]{algorithm2e}
\usepackage{listings}
\lstdefinelanguage{Rust}{
  sensitive,
  morecomment=[l]{//},
  morecomment=[s]{/*}{*/},
  moredelim=[s][{\itshape\color[rgb]{0,0,0.75}}]{\#[}{]},
  morestring=[b]{"},
  alsodigit={},
  alsoother={},
  alsoletter={!},
  otherkeywords={=>},
  morekeywords={break, continue, else, for, if, in, loop, match, return, while},
  morekeywords={as, const, let, move, mut, ref, static, unsafe},
  morekeywords={dyn, enum, fn, impl, Self, self, struct, trait, type, use, where},
  morekeywords={crate, extern, mod, pub, super},
  morekeywords={abstract, alignof, become, box, do, final, macro,
    offsetof, override, priv, proc, pure, sizeof, typeof, unsized, virtual, yield},
  morekeywords=[2]{Send},
  morekeywords=[3]{bool, char, f32, f64, i8, i16, i32, i64, isize, str, u8, u16, u32, u64, unit, usize, i128, u128},
}%

\usepackage{hyperref}

\usepackage{xspace}

\newcommand{\GJ}{\textsf{GJ}\xspace}
\newcommand{\EM}{\textsf{EM}\xspace}

\newcommand{\Egraph}{\mbox{E-graph}\xspace}
\newcommand{\egraph}{\mbox{e-graph}\xspace}

\newcommand{\egraphs}{\mbox{e-graphs}\xspace}

\newcommand{\eclass}{\mbox{e-class}\xspace}
\newcommand{\Eclasses}{\mbox{E-classes}\xspace}
\newcommand{\eclasses}{\mbox{e-classes}\xspace}
\newcommand{\ematch}{\mbox{e-match}\xspace}
\newcommand{\ematching}{\mbox{e-matching}\xspace}
\newcommand{\Ematching}{\mbox{E-matching}\xspace}
\newcommand{\Enodes}{\mbox{E-nodes}\xspace}
\newcommand{\enodes}{\mbox{e-nodes}\xspace}
\newcommand{\enode}{\mbox{e-node}\xspace}
\newcommand{\equivid}{\ensuremath{\equiv_{\sf id}}\xspace}
\newcommand{\find}{\textsf{find}\xspace}
\newcommand{\lookup}{\textsf{lookup}\xspace}
\newcommand{\egg}{\texttt{egg}\xspace}

\newcommand{\update}[1]{#1}

\usepackage[capitalize,nameinlink]{cleveref}

\def\drawplusplus#1#2#3{\hbox to 0pt{\hbox to #1{\hfill\vrule height #3 depth
      0pt width #2\hfill\vrule height #3 depth 0pt width #2\hfill
      }}\vbox to #3{\vfill\hrule height #2 depth 0pt width
      #1 \vfill}}
\newtheorem{theorem}{Theorem}

\begin{document}

\title{Relational E-matching}

\begin{abstract}
  %\egraphs good

% We present a new approach to e-matching
%   based on recent relational join techniques
%   from the database community.
We present a new approach to e-matching
  based on relational join;
  in particular,
  we apply recent database query execution techniques
  to guarantee worst-case optimal run time.
Compared to the conventional backtracking approach
  that always searches the e-graph ``top down'',
  our new \textit{relational e-matching} approach
  can better exploit pattern structure
  by searching the e-graph according to an optimized query plan.
We also establish the first
  data complexity result for e-matching,
  bounding run time as a function
  of the e-graph size and output size.
We prototyped and evaluated our technique in the
  state-of-the-art \egg e-graph framework.
Compared to a
  conventional baseline,
  relational e-matching is
  simpler to implement and
  orders of magnitude faster in practice.

%%% Local Variables:
%%% mode: latex
%%% TeX-master: "main"
%%% End:

\end{abstract}


\maketitle

\section{Introduction}

The congruence closure data structure,
 also known as the \egraph,
 is a central component of SMT-solvers~\cite{simplify, z3, moskal, cvc4}
 and equality saturation-based optimizers~\cite{eqsat, egg}.
An \egraph compactly represents a set of terms
 and an equivalence relation over the terms.
An important operation on \egraphs is \ematching,
 which finds the set of terms in an \egraph
 matching a given pattern.
In SMT-solvers, \ematching is used to instantiate
 quantified formulas over ground terms.
% For instance, suppose the solver input contains the
%  formula $\forall x . \phi(x, g(x)) \rightarrow \psi(x)$
%  for predicate symbol $\phi, \psi$
%  and function symbol $g$.
% The solver will search for potential instantiation
%  of the form $\phi(x, g(x))$ in its internal \egraph
%  data structure.
% For every ground term found,
%  the solver is able to deduce $\psi(x)$ as a fact.
In equality saturation, \ematching is used
 to match rewrite rules on an \egraph
 to discover new equivalent programs.
% In equality saturation, rewrite rules of the form
%  $f(\alpha, g(\alpha)) \Rightarrow h(\alpha)$ match
%  the left-hand-side pattern $f(\alpha, g(\alpha))$
%  on terms represented by the \egraph;
% upon every match, the right-hand-side $h(\alpha)$ is
%  instantiated with the substitution $\alpha \mapsto t$
%  into $h(t)$ and is added to the \egraph.
% One can view the rewrite rule
%  $f(\alpha, g(\alpha)) \Rightarrow h(\alpha)$
%  as a formula where the pattern variable $\alpha$ is
%  $\forall$-quantified,
%  so the \ematching~instances in SMT-solvers and equality
%  saturation are practically the same.
The efficiency of \ematching~greatly affects the overall performance
 of the SMT-solver~\cite{z3, cvc4},
 and slow \ematching~is a major bottleneck
 in equality saturation~\cite{egg, tensat, szalinski}.
In a typical application of equality saturation,
 \ematching~is responsible for 60--90\% of the overall run time~\cite{egg}.

Several algorithms have been proposed for e-matching~\cite{efficient-ematching, moskal, simplify}.
However, due to the NP-completeness of \ematching~\cite{ematching-nph},
 most algorithms implement some form of backtracking search,
 which are inefficient in many cases.
In particular, backtracking search only exploits \textit{structural constraints},
 which are constraints about the shape of a pattern,
 but defers checking \textit{equality constraints},
 which are constraints that variables should be consistently mapped.
This leads to suboptimal run time when
 the equality constraints dominate the structural constraints.

To improve the performance of backtracking-based \ematching,
 existing systems implement various optimizations.
Some of these optimizations only deal with
 patterns of certain simple shapes and are therefore \textit{ad hoc}
 in nature~\cite{moskal}.
Others attempt to incrementalize \ematching\ upon changes to the \egraph,
 or match multiple similar patterns together
 to eliminate duplicated work~\cite{efficient-ematching}.
However, these optimizations are complex to implement
 and fail to generalize to workloads where the \egraph~changes rapidly
 or when the patterns are complex.

% Finally, despite the NP-completeness of the general \ematching~problem, most
% researchers acknowledge \ematching~is almost always tractable, thanks to the
% small size of patterns encountered in practice. This hints at possible
% algorithmic improvements to \ematching.

To tackle the inefficiency and complexity involved in \ematching,
 we propose a systematic approach to \ematching~called \textbf{relational \ematching}.
Our approach is based on the observation that
 \ematching\ is an instance of a well-studied problem
 in the databases community,
 namely answering {\em conjunctive queries}.
We therefore propose to solve \ematching\ on an \egraph\ by reducing it
 to answering conjunctive queries on a relational database.
This approach has several benefits.
First, by reducing \ematching to conjunctive queries,
 we simplify \ematching by taking advantage of
 decades of study by the databases community.
Second, the relational representation provides a unified way to
 express both the structural constraints and equality constraints in patterns,
 allowing query optimizers to leverage both kinds of constraints
 to generate asymptotically faster query plans.
Finally, by leveraging the generic join algorithm,
 a novel algorithm developed in the databases community,
 our technique achieves the first worst-case optimal bound for \ematching.

Relational \ematching is provably optimal despite the NP-hardness of \ematching.
The databases community makes a clear distinction between
 {\em query complexity}, the complexity dependent on the size of the query,
 and {\em data complexity}, the complexity dependent on the size of the database.
The NP-hardness result~\cite{ematching-nph} is stated over the size of the pattern,
 yet in practice only small patterns are matched on a large \egraph.
When we hold the size of each pattern constant,
 relational \ematching runs in time polynomial over the size of the \egraph.

Our approach is widely applicable.
For example, \textit{multi-patterns} are
 typically framed as an extension to \ematching
 that allows the user to find matches
 satisfying multiple patterns simultaneously.
Efficient support for multi-patterns requires modifying the basic
 backtracking algorithm~\cite{efficient-ematching}.
In contrast, relational \ematching inherently supports multi-patterns for free.
The relational model also opens the door
 to entirely new kinds of optimizations,
 such as persistent or incremental \egraphs.

To evaluate our approach, we implemented relational \ematching for \egg, a
 state-of-the-art implementation of \egraphs.
Relational \ematching
 is simpler, more modular, and orders of magnitude faster
 than \egg's \ematching implementation.

In summary, we make the following contributions in this paper:
\begin{itemize}
\item We propose relational \ematching, a systematic approach to \ematching that
  is simple, fast, and optimal.
\item We adapt generic join to implement relational \ematching,
  and provide the first data complexity results for \ematching.
\item We prototyped relational \ematching\footnote{
  Our implementation can be found here: \url{https://github.com/yihozhang/relational-ematching-benchmark}.
  } in \egg,
  a state-of-the-art \egraph implementation,
  and we show that relational \ematching can be orders of magnitude faster.
\end{itemize}

The rest of the paper is organized as follows: \autoref{sec:background} reviews
relevant background on the \egraph~data structure, the \ematching~problem,
conjunctive queries and join algorithms.
\autoref{sec:algorithm} presents our
 relational view of \egraphs,
 our \ematching algorithm,
 and the complexity results.
\autoref{sec:optimization} discusses
optimizations on our core algorithm and addresses various practical concerns.
\autoref{sec:evaluation} evaluates our algorithm and implementation with a set of
experiments in the context of equality saturation.
\autoref{sec:discussion} discusses how the relational model opens up
 many avenues for future work in \egraphs and \ematching,
 and \autoref{sec:conclusion} concludes.

%%% Local Variables:
%%% mode: latex
%%% TeX-master: "main"
%%% End:

\section{Background}\label{sec:background}

Throughout the paper we follow the notation in \autoref{fig:notation}.
%This section introduces the concepts in the figure.
We define the
 \egraph\ data structure and the \ematching\ problem,
 and review background on
 relational queries and join algorithms
 that form the foundation of our \ematching\ algorithm.

\subsection{E-Graphs and E-Matching}\label{sec:background:ematching}

\newcommand\nt[1]{\ensuremath{\langle\mathit{#1}\rangle}}
\begin{figure}
  \centering
  \begin{tabular}{llcl}
     function symbols & \nt{fun} & ::= & $f \mid g \mid \ldots$                                           \\
     variables        & \nt{var} & ::= & $x \mid y \mid z \mid \ldots \mid \alpha \mid \beta \mid \ldots$ \\
     \eclass ids      & \nt{id}  & ::= & $i \mid j \mid \ldots$                                           \\
     ground terms     & \nt{t}   & ::= & $\nt{fun} \mid \nt{fun}\left(\nt{t},\ldots,\nt{t}\right)$        \\
     patterns         & \nt{p}   & ::= & $\nt{fun} \mid \nt{fun}(\nt{p},\ldots,\nt{p}) \mid \nt{var}$     \\
     \enodes          & \nt{n}   & ::= & $\nt{fun} \mid \nt{fun}(\nt{id}, \ldots, \nt{id})$               \\
     \eclasses        & \nt{c}   & ::= & $\{ \nt{n}, \ldots, \nt{n}\}$                                    \\
  \end{tabular}
  \caption{Syntax and metavariables used in this paper.}
  \label{fig:notation}
\end{figure}

Let $\Sigma$ be a set of function symbols with associated arities. A function
symbol is called a \textit{constant} if it has zero arity. Let $V$ be a set
of variables. We define $T(\Sigma, V)$ to be the set of terms constructed using
function symbols from $\Sigma$ and variables from $V$:

\update{
\begin{definition}[Terms and patterns]
  The set of terms $T(\Sigma,V)$ over function symbols $\Sigma$ and variables $V$
 is the smallest set such that (1) all variables and constants are in
$T(\Sigma,V)$ and (2) $t_1,\dots,t_k\in T(\Sigma,V)$ implies
$f(t_1,\dots,t_k)\in T(\Sigma,V)$, where $f\in \Sigma$ has arity $k$. A
\textit{ground term} is a term in $T(\Sigma,V)$ that contains no variables.
All terms in $T(\Sigma, \emptyset)$ are ground terms.
A non-ground term is also called a \textit{pattern}. We call a term of the form
$f(t_1,\ldots,t_k)$ an $f$-application term.
\end{definition}

We define a {\em congruence relation} over the terms as an equivalence relation 
 that is {\em congruent}:

\begin{definition}[Equivalence Relation]
  An \textit{equivalence relation} $\equiv_\Sigma$ is a binary relation over
$T(\Sigma,\emptyset)$ that is reflexive, symmetric, and transitive. 
\end{definition}

\begin{definition}[Congruence Relation]
  A \textit{congruence relation} $\cong_\Sigma$ is an equivalence relation
satisfying:
\[
  \forall \text{$k$-ary function symbols $f$}.\
  \left(
    \forall i\in\{1, \ldots, k\}.\ t_i \cong t'_i
  \right)
  \implies
  f(t_1,\dots,t_k)\cong f(t'_1,\dots,t'_k)
\]  
The {\em congruence closure} of a binary relation $R$ on $T(\Sigma, \emptyset)$
 is the smallest congruence relation that contains $R$.
\end{definition}

We write $\equiv$ and $\cong$ when $\Sigma$ is clear from the context.

\newcommand{\child}{\textit{child}} \newcommand{\sym}{\textit{symbol}}

An \egraph is a data structure that represents a congruence relation.
An \egraph is built up from {\em \eclasses} and {\em \enodes}, defined as follows:
\begin{definition}[\Eclasses and \Enodes]
  An \eclass is a set of \enodes. Every \eclass is identified by one or more ids.
  An \enode is a tuple $(f, \text{args})$ where $f$ is a function symbol 
   and args is a (possibly empty) list of \eclass ids. 
Similar to terms, we call an \enode\ of the form $(f, i_1, \ldots i_k)$
 an $f$-application \enode.
We will write $f(i_1, \ldots, i_k)$ for the \enode $(f, i_1, \ldots i_k)$
\end{definition}

\begin{definition}
  A \egraph is a tuple $(C, E, I, U, M, \lookup)$ where:
  \begin{itemize}
    \item $C$ is a set of \eclasses over $E$ which is a set of \enodes;
    $I$ is a set of \eclass ids.
    \item
      A union-find~\cite{unionfind} data structure $U$
      stores an equivalence relation (denoted with $\equivid$)
      over \eclass ids.
      The union-find provides a function \find that
       \textit{canonicalizes} \eclass ids such that
       $\find(i_1) = \find(i_2) \iff i_1 \equivid i_2$.
      An \eclass id $i$ is canonical if $i = \find(i)$.
    
    \item
      The \textit{\eclass map} $M$ is a surjective function 
      that maps \eclass ids to \eclasses.
      All equivalent \eclass ids map to the same \eclass, i.e.,
      $a \equivid b$ iff $M[a]$ is the same set as $M[b]$.
      An \eclass id $a$ is said to \textit{refer to} the \eclass $M[\find(a)]$.
    
    \item
      A function \lookup that maps \enode $n$ to an id of
        \eclass that contains it: $n \in M[\lookup(n)]$.
    \end{itemize}
\end{definition}
}
 
Note that by definition, no two \enodes in the same \egraph 
 can have the same symbol and children, i.e.,
 an \enode's symbol and children together uniquely identify
 the \eclass that contains it.
This property is necessary for \lookup to be a function.
\autoref{sec:fd} explains how this property also
 translates to a {\em functional dependency}
 in the \egraph's relational representation,
 which could be leveraged to further optimize relational \ematching.

 \update{
\autoref{fig:egraph-eg} shows an example \egraph, where each dotted box
 is an \eclass, and each solid box together with argument pointers 
 make up an \enode. 
In this \egraph there is one unique id for each \eclass,
 shown in the shaded labels on the top-left corner of each \eclass. 
For example, the \eclass in the middle contains the set of \enodes 
$g(1), \ldots, g(N)$ and has id $i_g$. }

An \egraph $E$ efficiently represents sets of ground terms in a congruence
 relation.
An \egraph, \eclass, or \enode is said to \textit{represent} a term $t$ if $t$ can be
``found'' within it. 

\begin{definition}[Representation]
  Representation is defined recursively:
\begin{itemize}
\item An \egraph\ represents a term if any of its \eclasses\ does.
\item An \eclass\ $c$ represents a term if any \enode $n \in c$ does.
\item An \enode $f(j_{1}, \ldots, j_k)$ represents a term $f(t_{1}, \ldots, t_k)$
  if they have the same function symbol $f$
  and \eclass $M[j_{i}]$ represents term $t_{i}$ for $i \in \{1, \ldots, k\}$.
\end{itemize}
\end{definition}

%  An \egraph~is said to \textit{represent} a ground term $t$ if its
% \eclasses~represent it. An \eclass $c$ represents a ground term if $c$ contains
% an \enode $a$ that represents it. An \enode $f(c_1,\dots,c_k)$ represents a
% ground term $f(t_1,\dots,t_k)$ if they have the same function symbol $f$ and
% each e-class $c_i$ represents term $t_i$. Terms represented by an \eclass~are
% equivalent to each other.

\begin{figure}
  \centering
  \includegraphics[height=5cm]{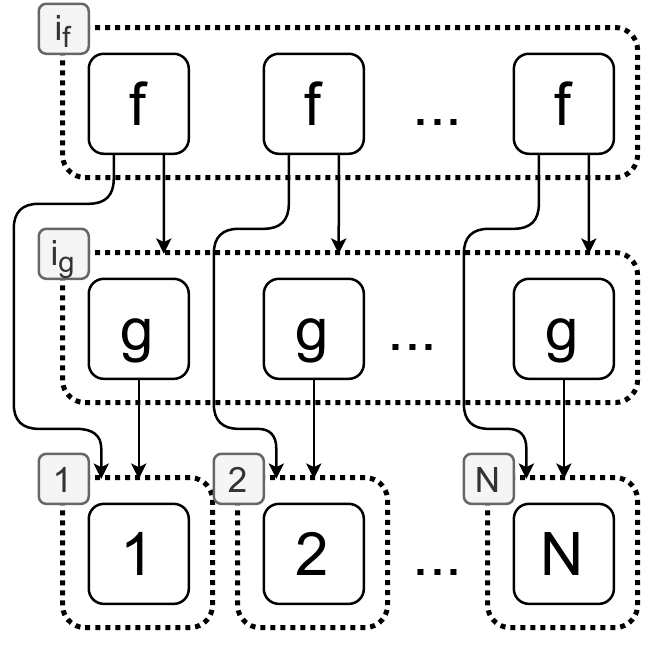}
  \caption{An example \egraph.}\label{fig:egraph-eg}
\end{figure}

The \egraph in~\autoref{fig:egraph-eg} represents the set of
terms (where $[N] = \{1, 2, \ldots, N\}$):
$$[N]\cup \{g(i) \mid i\in [N]\} \cup\{f(i, g(j)) \mid i, j \in [N]\}.$$
In addition, all $g$-terms are equivalent,
 and all $f$-terms are equivalent.
Note that the \egraph~has size $O(N)$,
 yet it represents $\Omega(N^{2})$ many terms.
In general, an \egraph is capable of representing
 exponentially many terms in polynomial space.
If the \egraph has cycles,
 it can even represent an infinite set of terms.
For example, the \egraph with a single \eclass $c = \{f(c), a\}$
 represents the infinite set of terms $\{a, f(a), f(f(a)), \ldots\}$.

\subsubsection*{\Ematching.}
\Ematching~finds the set of terms in an \egraph~matching a given pattern.
Specifically, \ematching~finds the set of \ematching~substitutions and a root
class that represents the terms.

\begin{definition}[\Ematching substitution]
  An \textit{\ematching~substitution} $\sigma$ is a function that maps
every variable in a pattern to an \eclass.  
\end{definition}
 For convenience, we use $\sigma(p)$ to
denote the set of terms obtained by replacing every occurrence of variable $v_i$
in $p$ with terms represented by $\sigma(v_i)$.

\begin{definition}[The \Ematching problem]
  Given an e-graph $E$ and a pattern $p$,
\ematching~finds the set of all possible pairs $(\sigma, r)$ such that every
term in $\sigma(p)$ is represented in the \eclass $r$. Terms in $\sigma(p)$ are
said to be matched by pattern~$p$, and $r$ is said to be the root of matched terms.  
\end{definition}
For example, matching the pattern $f(\alpha, g(\alpha))$ against the \egraph~$G$
in \autoref{fig:egraph-eg} produces the following $N$ substitutions,
 each with the same root $i_f$:
\[
  \left\{
    (\{\alpha \mapsto j\}, i_f) \mid j \in [N]
  \right\}.
\]

Existing \ematching~algorithms perform backtracking search directly on the
\egraph~\cite{efficient-ematching,simplify,egg}.
\autoref{alg:backtrack} shows
 an abstract backtracking-based
 \ematching~algorithm.
Most
\ematching~algorithms using backtracking search can be viewed as
optimizations based on this abstract algorithm. Specifically,
it will perform a top-down search following the shape of the pattern and prune
the result set of substitutions when necessary.
To match pattern $f(\alpha,
g(\alpha))$ against $G$, backtracking search visits terms in the following order
(each $\hookrightarrow$ marks a backtrack step):
% TODO: Say something like f(1, g(1)) is not actually a term, but a (singleton)
% set of terms witnessed by the substitution {alpha -> 1}.
\begin{align*}
  f(1, g(1)) \rightarrow & \dots \rightarrow f(1, g(N)) \\
    \hookrightarrow f(2, g(1))\rightarrow & \dots \rightarrow f(2, g(N)) \\
    \hookrightarrow f(N, g(1))\rightarrow & \dots \rightarrow f(N, g(N))
\end{align*}

For each term $f(i, g(j))$ visited, whenever $i=j$ the algorithm yields a match
$\alpha \mapsto i$. Despite there being only $N$ matches,
backtracking search runs in time $O(N^{2})$.
% Existing approaches to \ematching~rely on backtracking. For example,
% \citet{efficient-ematching} proposed a backtracking-based e-matching algorithm
% that is used by \textsc{Z3} \cite{Z3} and \egg \cite{egg}, two
% state-of-the-art e-graph implementations. To match the pattern $f(\alpha,
% g(\alpha))$ on the e-graph in Figure \ref{fig:egraph}, their algorithm does a
% depth-first search over the e-graph:
% % Conceptually, it recursively searches for the pattern $f(\alpha, g(\beta))$
% % and filters out substitutions that map $\alpha$ and $\beta$ to different
% % e-classes.
% it searches for all $f$-application e-nodes $n_f$, adds $\alpha\mapsto
% n_f.\textit{child}_1$ to substitution $\sigma$, iterates through all
% $g$-application e-nodes $n_g$ in e-class $n_f.\textit{child}_2$, and only
% yield $\sigma$ if $n_g.\textit{child}_1=\sigma(\alpha)$. In general, this
% procedure runs in time that is quadratic of the \egraph~size. In a large
% e-graph, there may be thousands of pairs of $n_f$ and $n_g$ where $n_g$ is in
% e-class $n_f.\textit{child}_2$, but only a few satisfy the constraint
% $n_f.\textit{child}_1=n_g.\textit{child}_1$.

This inefficiency is due to the fact that na\"ive backtracking does not use the
equality constraints to prune the search space \textit{globally}. Specifically,
the above \ematching~pattern corresponds to three constraints for a potential
matching term $t$:
\begin{enumerate}
\item $t$ should have function symbol $f$.
\item $t$'s second child should have function symbol $g$.
\item $t$'s first child should be equivalent
      to the child of $t$'s second child.
\end{enumerate}

We can categorize these constraints into two kinds:
\begin{itemize}
  \item \textit{Structural constraints} are derived from the structure of the pattern.
    The structure of pattern
    $f(\alpha, g(\alpha))$ constrains the root symbol and the second symbol
    to be $f$ and $g$ respectively (i.e., constraints 1 and 2).
  \item \textit{Equality constraints}
    are implied by multiple occurrences of the same variable.
    Here, the occurrences of $\alpha$ implies that the terms at these positions
    should be equivalent with each other for all matches (i.e., constraint 3),
     which we call {\it equality constraints}.
    Following \citet{moskal}, we define patterns without equality constraints to be
    \textit{linear patterns}.
\end{itemize}

Backtracking search exploits the structural constraints first and defers
 checking the equality constraints to the end.
In our example pattern $f(\alpha, g(\alpha))$,
 backtracking search enumerates all $f(i, g(j))$,
 regardless of whether $i$ and $j$ are equivalent,
 only to discard nonequivalent matches later.
Complex query patterns may involve many variables
 that occur at several places, which will makes na\"ive backtracking search
 enumerate a very large number of candidates,
 even though the result size is small.
% In contrast, our approach exploits all equality constraints during query planning,
%  and can yield asymptotically better performance for non-linear patterns
%  and guarantees worst-case optimal run time in general.

\def\match{{\sf match}} \def\class{{\sf class}} \def\find{{\sf find}}
\def\dom{{\sf dom}}
\begin{figure}
  \begin{align*}
    \match(x,c,S) = & \{ \sigma \cup \{ x \mapsto c\} \mid \sigma \in S, x \not \in \dom(\beta)\}\ \cup\\
                    & \{ \sigma \mid \sigma \in S, \sigma(x) = c \}\\
                    % \match(c,t,S) = & c \in \class(t) \text{ ? } S \text{ : } \emptyset \\
    \match(f(p_{1}, \dots, p_{k}), c, S) = & \bigcup_{f(c_{1},\dots,c_{k})\in c}
                                             \match(p_{k}, c_{k}, \dots, \match(p_{1}, c_{1}, S))
  \end{align*}
  \caption{A declarative backtracking-based \ematching{} algorithm (reproduced
    from~\citet{efficient-ematching}). The set of substitutions for pattern $p$
    on \egraph~$G$ with \eclasses $C$ can be obtained by computing $\bigcup_{c\in C}\match(p, c, \emptyset)$.
    % \textbf{\color{red} maybe lets use an algo presentation to match up with GJ}
  }\label{alg:backtrack}
\end{figure}

\subsection{Conjunctive Queries}

Conjunctive queries are a subset of queries in relational algebra
 that use only select, project, and join operators
(as opposed to union, difference, or aggregation).
Conjunctive queries have many desirable theoretical properties
 (like computable equivalence checking),
 and they enjoy efficient execution
 thanks to decades of research from the databases community.

\subsubsection*{Relational Databases.}
A relational schema $S_D$ over domain $D$ is a set of relation symbols with
 associated arities.
A relation $R$ under a schema $S_D$ is a set of tuples;
 for each tuple $(t_1,\ldots,t_k) \in R$,
 $k$ is the arity of $R$ in $S_D$ and $t_i$ is an element in $D$.
A database instance (or simply database) $I$ of $S_D$ is a set of relations under $S_D$.

We use the notation $R(x, y).x$ to denote projection,
 i.e., $R(x, y).x = \{x \mid (x, y) \in R \}$.

\subsubsection*{Conjunctive Queries.}
A conjunctive query $Q$ over the schema $S_D$ is a formula of the form:
\[
  Q(x_1,\ldots x_k) \gets R_1(x_{1,1},\ldots,x_{1,k_1}),\ldots,
  R_n(x_{n,1},\ldots,x_{n,k_n}),
\]
where $R_1\ldots R_n$ are relation symbols in $S_D$ with arities
$k_1,\ldots k_n$ and the $x$ are variables.\footnote{
  Some definitions of conjunctive queries allow both variables and constants.
  We only allow variables without loss of generality:
  any constant $c$ can be specified with a distinguished relation $R_{c} = \{c\}$.}
We call the $Q(\ldots)$ part the \textit{head} of the query,
 the remainder is the \textit{body}.
Each $R_i(\ldots)$ is called an {\em atom}.
Variables that appear in the head are called \textit{free variables},
 and they must appear in the body.
Variables that appear in the body but not the head are called \textit{bound variables},
 since they are implicitly existentially quantified.

% $v_{i,j}$ are variables in $V$.
% and $x_1,\ldots x_k$ are variables
% occurring in $v_{i,j}$.

\subsubsection*{Semantics of Conjunctive Queries.}
Similar to \ematching,
 evaluating a conjunctive query $Q$ yields substitutions.
Specifically,
 evaluation yields substitutions
 that map free variables in $Q$ to
 elements in the domain such that
 there exists a mapping of the bound variables
 that causes every substituted atom to be present in the database.
Bound variables are {\em projected out} and not present in
 resulting the substitutions.

% \textbf{\color{blue} This definition may cause problems if variables in
%   the body do not occur in the head.
% Instead, we may say the result of CQ is a set of tuples, instantiated from a
% substitution (of all variables in the body) such
% that ... Example:
% [...] Evaluating $Q$ over $I$ yield the set of all possible
%  $(\sigma(x_1), \ldots, \sigma(x_n))$,
%  where $\sigma$ is a substitution satisfying
%  \[
%    \forall j\in [n].\
%    R_j\left( \sigma(x_{j,1}),\ldots, \sigma(x_{j,k_j})\right) \in I.
%  \]
% }

More formally,
 let $I$ be a database of schema $S_D$ and
 let $Q$ be a conjunctive query over the same schema
 with $k$ variables in its head.
Let the $n$ atoms in the body of
 $Q$ be $R_1, \ldots, R_n$
 where $R_j$ has arity $k_j$.
Evaluating $Q$ over $I$ yields
 a substitution
 $\sigma = \{ x_1 \mapsto t_1, \ldots, x_k \mapsto t_k \}$
 iff
 there exists a $\sigma' \supset \sigma$ mapping
 all variables in $Q$ such that:
\[
  \bigwedge_{j \in [n]}(\sigma'(x_{j,1}),\ldots,\sigma'(x_{j,k_j})) \in R_{j}
\]
% Evaluating $Q$ over $I$ yields substitutions
%  % $\sigma = \{ x_1 \mapsto t_1, \ldots, x_k \mapsto t_k \}$
%  $\sigma = \{ x_1 \mapsto \sigma'(x_1), \ldots, x_k \mapsto \sigma'(x_k) \}$,
%  where $\sigma'$ is a substitution subsuming $\sigma$
%  that map variables from the bodies of $Q$ to elements from $D$.
% Specifically, $\sigma'$ satisfies the following:
% \[
%   \forall j \in [n].\
%   \left( \sigma'(x_{j,1}), \ldots, \sigma'(x_{j,k_j}) \right) \in R_j
% \]
% In other words, $Q$ returns the set (where $x_{l}, \ldots, x_{m}$ are body variables not in the head):
%
%
%  $\sigma = \{ x_1 \maps t_1, \ldots, x_k \maps t_k \}$
%  that map all $k$ variables
%  in the head of $Q$ to elements of $D$ such that

% $R_i(\sigma(t_1),\ldots \sigma(t_k))$ are atoms in
%  $(\sigma(x_1),\ldots,\sigma(x_k))$ where $\sigma$ is a
%  substitution satisfying that $R_i(\sigma(t_1),\ldots \sigma(t_k))$ are atoms in
% $I$ for $i=1,\ldots,n$.
% We denote the result set as $Q(I)$.

% Similar to \ematching~substitutions, we define a conjunctive query substitution
% as a function that maps every variable occurring in it to elements in
% $D$. The semantics of conjunctive queries is defined as follows: let $Q$ be a
% conjunctive query, $I$ be an database of schema of $S_D$.
% Conjunctive query $Q$ yields all tuples
% set of all $\textit{ans}(\sigma(x_1),\ldots,\sigma(x_k))$ where $\sigma$ is a
% substitution satisfying that $R_i(\sigma(t_1),\ldots \sigma(t_k))$ are atoms in
% $I$ for $i=1,\ldots,n$. We denote the result set as $q(I)$.

In practice, conjunctive queries are often evaluated according to a \textit{query plan}
 which dictates each step of execution.
For example,
 many industrial database systems will construct tree-like query plans,
 where each node describes an operation like scanning a relation or joining two
 intermediate relations.
Industrial database systems typically construct query plans
 based on binary join algorithms such as hash joins and merge-sort join,
 which process two relations at a time.
The quality of a query plan critically
 determines the performance of evaluating a conjunctive query.

We observe that conjunctive query and \ematching
 are structurally similar: both are defined as finding
 substitutions whose instantiations are present in a database.
% In fact, the \ematching problem can be viewed as ``nested''
%  conjunctive queries on a database where the relations are nested.
Therefore, it is tempting to reduce \ematching to
 a conjunctive query over the relational database,
 thereby benefiting from well-studied techniques from the databases community,
 including join algorithms and query optimization.
We achieve exactly this in \autoref{sec:algorithm}.

\subsection{Worst-Case Optimal Join Algorithms}

% \begin{figure}
%   \begin{subfigure}{0.45\textwidth}
%     \centering \tikzfig{worstcase}
%     \caption{ Input relations. }\label{fig:worst}
%   \end{subfigure}
%   \begin{subfigure}{0.45\textwidth}
%     \centering \tikzfig{triangle}
%     \caption{ $Q \leftarrow R(x,y), S(y,z), T(z, x)$. }\label{fig:triangle}
%   \end{subfigure}
%   \caption{ Example input and query hypergraph of $Q$. }\label{fig:wcoj}
% \end{figure}

% \begin{figure}
% \begin{align*}
% \textsf{GJ}(\bigwedge_i R_i(\overline{x_i}), [x_1] \concat \overline{x}, \sigma) & =
% \{ \textsf{GJ}(x_1 = a \wedge \bigwedge_i R_i(\overline{x_i}), \overline{x}, \sigma \cup \{ x_{1} \mapsto a \} ) \mid a \in \bigcap_{x_1 \in R_j} R_j \} \\
% \textsf{GJ}(\bigwedge_i R_i(\overline{x_i}), \emptyset, \sigma) & = \sigma
% \end{align*}
% \caption{Generic Join algorithm. \textsf{GJ} takes 3 arguments:
%  a conjunction of atoms $\bigwedge_i R_i(\overline{x_i})$,
%  a list of variables indicating the variable ordering,
%  and a (partial) substitution $\sigma$.}
% \end{figure}

The run time of any algorithm for answering conjunctive queries
 is lower-bounded by the output size,
 assuming the output must be materialized.
How large can the output of a conjunctive query be on a particular database?
A na{\"i}ve bound is simply the product of the size of each relation,
 which is the size of their cartesian product.
Such a na{\"i}ve bound fails to consider the query's structure.
The AGM bound~\cite{agm}
 gives us a bound for the worst-case output size.
In fact, the AGM bound is tight;
 there always exists a database
 where the query output size is the AGM bound.

The AGM bound and worst-case optimal joins
 are recent developments in databases research.
We do not attempt to provide a comprehensive
 background on these topics here;
 familiarity with the existence of the AGM bound
 and the generic join algorithm is sufficient for this paper.

Consider $Q(x,y,z) \leftarrow R(x,y), S(y,z), T(x,z) $,
 also known as the ``triangle query'',
 since output tuples are triangles between the edge relations $R,S,T$.
We calculate a trivial bound $|Q| \leq |R| \times |S| \times |T|$.
If $|R|=|S|=|T|=N$, then $|Q| \leq N^3$.
We can derive a tighter bound from $|Q| \leq |R| \times |S| = N^2$.
That is because $Q$
contains fewer tuples than the query $ Q’(x,y,z) \leftarrow R(x,y), S(y,z) $
as $Q$ further requires $(x,z) \in T$.
The AGM bound for $Q$ is even smaller: $N^{3/2}$.
It is computed from the \emph{fractional edge
  cover} of the \emph{query hypergraph}.

\subsubsection*{Query Hypergraph.}
\label{sec:hypergraph}

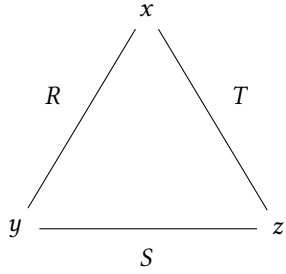
\begin{figure}
  \centering \begin{tikzpicture}
	\begin{pgfonlayer}{nodelayer}
		\node [style=node] (3) at (3.5, -2) {$z$};
		\node [style=node] (4) at (0, 3.75) {$x$};
		\node [style=node] (5) at (-3.5, -2) {$y$};
		\node [style=none] (10) at (2.5, 1.5) {$T$};
		\node [style=none] (11) at (0, -2.75) {$S$};
		\node [style=none] (12) at (-2.5, 1.5) {$R$};
	\end{pgfonlayer}
	\begin{pgfonlayer}{edgelayer}
		\draw (4) to (3);
		\draw (5) to (3);
		\draw (4) to (5);
	\end{pgfonlayer}
\end{tikzpicture}
  \caption{ Query hypergraph of $Q(x, y, z) \gets R(x,y), S(y,z), T(z, x)$. }\label{fig:wcoj}
\end{figure}

The hypergraph of a query is simply the hypergraph
 with
 a vertex for each variable
 and a (hyper)edge for each atom.
The edge for an atom $R(x_{i}, \ldots, x_{k})$ connects the vertices
 corresponding to the variables $x_{i}, \ldots, x_{k}$.
\autoref{fig:wcoj} illustrates $Q$'s hypergraph.
% Certain hypergraphs are called \emph{cyclic}.
% The definition of a hypergraph cycle is related to but not exactly
%  that of cycles in conventional graphs.

\subsubsection*{Cyclic and Acyclic Queries.}
Certain queries can be represented by a tree,
 called the \textit{join tree},
 where each node corresponds to an atom in the query.
Furthermore, for each variable $x$ the nodes
 corresponding to the atom containing $x$
 must form a connected subtree.
Queries that admit such a join tree
 are said to be \textit{acyclic};
 otherwise, the query is \textit{cyclic}.\footnote{
  A cycle in the hypergraph does not necessarily entail a cyclic query,
   since the hypergraph may still admit a join tree.
}
The triangle query is cyclic because it cannot
 be represented by a join tree.
Acyclic queries can be answered more efficiently
 than cyclic ones.

\subsubsection*{Fractional Edge Cover.}
A set of edges \emph{cover} a graph if they touch all vertices. For $Q$'s
hypergraph, any two edges form a cover.
A \emph{fractional edge cover} assigns a
 weight in the interval $[0, 1]$ to each edge
 such that,
 for each vertex $v$,
 the weights of the edges containing $v$ sum to at least 1.
Every edge cover is a fractional cover,
 where every edge is assigned a weight of 1 if it is in the cover,
 and 0 otherwise.
For $Q$'s hypergraph,
 $\{R \mapsto 1/2, S \mapsto 1/2, T \mapsto 1/2\}$
 is the fractional edge cover with lowest total weight.

\subsubsection*{The AGM Bound.}
The AGM bound~\cite{agm} for a query with body atoms
 $R_{i}(\ldots)$ for $i \in [k]$
 is defined as $\min_{w_1,\ldots,w_k} \Pi_{i \in [k]} |R_{i}|^{w_i}$,
 where $\{R_{i} \mapsto w_{i} \mid i \in [k] \}$ forms a fractional edge cover.
For example, the AGM bound for $Q$ is
 $|R|^{1/2}|S|^{1/2}|T|^{1/2} = N^{3/2}$ when $|R| = |S| = |T| = N$.
This is the upper bound of $Q$’s output size;
 i.e.\ in the worst case $Q$ outputs this many tuples.

\subsubsection*{Generic Join.}

A desirable algorithm for answering conjunctive queries
 should run in time linear to the worst case output size.
Recent developments in the databases community have led to
 such an algorithm~\cite{wcoj}, one of which is \textit{generic join}~\cite{gj}.
Generic join has one parameter:
 an ordering of the variables in the query.
Any ordering guarantees a run time linear to the worst-case output size,
 but different orderings can lead to dramatically different run time
 in practice~\cite{emptyheaded}.

\begin{algorithm}
  \LinesNumbered{} \SetAlgoLined{}
  \KwResult{$\text{GJ}(Q, \emptyset)$ computes the output of query $Q$}
  \KwIn{query $Q$, partial substitution $\sigma$}
  \tcc{$k$ indicates how many variables remain in query $Q$}
  \eIf(\tcc*[f]{there are no more variables, so $\sigma$ is complete}){$k=0$}
    {output $\sigma$}
    {
      choose a variable $x$\;
      \tcc{
        Compute $D_x$, which all possible values of $x$, by intersecting the attributes of the relations where $x$ occurs.
        Intersection must be computed in $O(\min(|R_j.x|))$ time.}
      $J = \{j \mid x \in \overline{X}_j\}$\;
      $D_x = \cap_{j \in J} R_j.x$\label{ln:gj-isect}\;
      \For()
      {$v \in D_x$}
      {
        \tcc{compute residual query by replacing variable $x$ with constant $v$}
        $Q' = Q[v/x]$\;
        GJ$(Q', \sigma \cup \{ x \mapsto v \})$
      }
    }
  \caption{
    Generic join for the general query
    $Q(x_1,\ldots,x_k) \gets R_1(\overline{X}_1),\ldots,R_n(\overline{X}_n)$
  }\label{alg:gj-general}
\end{algorithm}

\autoref{alg:gj-general} shows the generic join algorithm.
Generic join is recursive,
 proceeding in two steps.
First,
 it chooses a variable from the query and
 collects all possible values for that variable in the query.
Then,
 for each of those values,
 it builds a \textit{residual query}
 by replacing occurrences of the variable in the query with a possible
 value for that variable.
These residual queries are solved recursively,
 and when there are no more variables in the residual query,
 the algorithm yields the substitution it has accumulated so far.

Generic join
 requires two important performance
 bounds to be met in order for its own run time to meet the AGM bound.
First, the intersection on line \ref{ln:gj-isect}
 must run in $O(\min(|R_j.x|))$ time.
Second,
 the residual relations
 should be computed in constant time,
 i.e., computing from the relation $R(x, y)$
 the relation $R(v_x, y)$ for some $v_x \in R(x, y).x$
 must take constant time.
Both of these
 can be solved by using tries
 (sometimes called prefix or suffix trees)
 as an indexing data structure.
\autoref{fig:simple-trie} shows an example trie.
Tries allow fast
 intersection because
 each node is a map which
 can be intersected in time linear to the size of the smaller map.
Tries also allow constant-time access
 to residual relations according to a compatible variable ordering.

 \begin{figure}
   \begin{subfigure}[b]{0.45\linewidth}
     \centering
     \begin{tabular}{|lll|}
       \hline
       \bf x & \bf y & \bf z \\
       \hline
       1 & 2 & 4 \\
       1 & 2 & 6 \\
       1 & 3 & 7 \\
       8 & 2 & 4 \\
       \hline
     \end{tabular}
     \caption{Table for relation $R(x, y, z)$.}
   \end{subfigure}
   \hfill
   \begin{subfigure}[b]{0.5\linewidth}
     \centering
     \includegraphics[height=22mm]{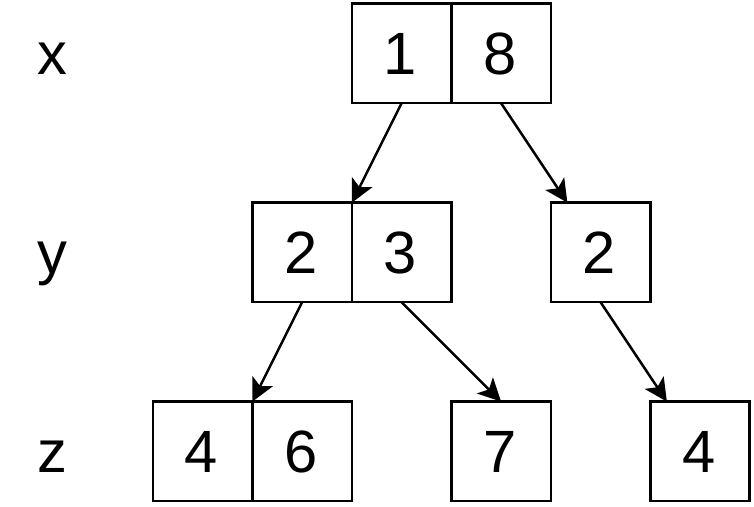}
     \caption{Trie for relation $R(x, y, z)$ using ordering $[x, y, z]$.}
   \end{subfigure}
   \caption{
     A trie is a tree where
       every node is a map (typically a hashmap or sorted map)
       from a value to a trie.
     Every path from the root of a trie to a leaf represents a tuple in the relation.
     Tries allow efficient computation of residual relations.
     For example, $R(1, y, z)$ can be computed quickly by following the
       $x=1$ edge from the root.
   }
   \label{fig:simple-trie}
 \end{figure}

A useful way to understand generic join is to observe the
 loops and intersections it performs on a specific query.
\autoref{alg:gj} shows generic join computing the triangle query.
Given a variable ordering ($[x,y,z]$ in this case),
 generic join assumes the input relations are stored in tries according to the ordering,
 so $R(x,y)$ is
 stored in a trie with $x$s on the first level, and $y$s on the second.
This makes accessing the residual relations ($R(v_x, y)$) fast since
 the replacement of variables with values
 is done according to the given variable ordering.
Note how the algorithm is essentially just nested \texttt{for} loops.
There is no explicit filtering step;
 the intersection of residual queries guarantees
 that once a complete tuple of values is selected,
 it can be immediately output without additional checking.

\begin{algorithm}
  \LinesNumbered{} \SetAlgoLined{}
  \KwResult{compute $Q(x,y,z) \gets R(x,y),S(y,z),T(z,x)$}
  $X = R(x,y).x \cap T(z,x).x$\;
  \For(\tcc*[f]{compute $Q(v_x,y,z) = R(v_x,y),S(y,z),T(z,v_x)$})
  {$v_x \in X$}
  {
    $Y = R(v_x,y).y \cap S(y,z).y$\;
    \For(\tcc*[f]{compute $Q(v_x,v_y,z) = R(v_x,v_y),S(v_y,z),T(z,v_x)$})
    {$v_y \in Y$}
    {
      $Z = S(v_y,z).z \cap T(z,v_x).z$\;
      \For(\tcc*[f]{yield join results $Q(v_x, v_y, v_z)$})
      {$v_z \in Z$}
      {output$(v_x,v_y,v_z)$ } } }
  \caption{
    Generic join for the triangle query, with ordering $[x, y, z]$.
  }\label{alg:gj}
\end{algorithm}

% \subsubsection{When to Use Generic Join.}

% Use Generic Join when the query is cyclic - if the query is acyclic, the
% theoretical optimal algorithm is Yannakakis (equivalent to message passing in
% probabilistic graphical models) which runs in linear time in the input size;
% in practice people use binary joins for acyclic queries, because Yannakakis
% always incurs a constant factor of 3 (it’s a 3-pass algorithm).

% In practice, Generic Join may have good performance compared to binary joins,
% since it is sometimes equivalent to a linear join tree. The “average case”
% complexity of Generic Join is open.

%%% Local Variables:
%%% mode: latex
%%% TeX-master: "main"
%%% End:

\section{Relational E-Matching}
\label{sec:algorithm}

% {\centering
% \begin{minipage}{.8\linewidth}
% \begin{algorithm}%[H]
%   \LinesNumbered{} \SetAlgoLined{} \KwIn{a function symbol $f$ and variable
%     ordering $\pi$. } \KwOut{ a trie $R_{f} : \{ \pi((c, \bar x)) \mid f(\bar
%     x) \in G[c] \}$. } \For{ $f(\bar x) \in G$ }{ $R_{f}$.insert$(\pi(f.\eclass,
%     \bar x))$ }
%   \caption{ Converting $f$-nodes in $G$ into a relation (as a trie).
%     \yz{should we keep this?}
%   }\label{alg:g2trie}
% \end{algorithm}
% \end{minipage}
% \par}

\Ematching via
 backtracking search is inefficient because it handles equality constraints
 suboptimally.
In fact, backtracking search follows edges in the \egraph and only
 visits concrete terms that satisfy structural constraints.
However, equality constraints
 are checked \emph{a posteriori} only after the search visits a (partial) term.\footnote{
   Backtracking \ematching can perform the check as soon as it has traversed
   enough of the pattern to encounter a variable more than once.
}
Whenever there are many
 terms that satisfy the structural constraints but not the equality constraints,
 as is in our example pattern $f(\alpha, g(\alpha))$,
 backtracking will waste time visiting terms that do not yield a match.

By reducing \ematching to evaluating conjunctive queries,
 we can use join algorithms that take advantage of both structural
 \emph{and} equality constraints.
\autoref{fig:emvhj} conveys this intuition
 using the pattern $f(\alpha, g(\alpha))$ and the
 example \egraph and database from \autoref{fig:egraph-and-table}.
The backtracking approach considers every possible assignment
 to the variables, even those where the two occurrences of $\alpha$
 do not agree.

We can instead formulate a conjunctive query
 that is equivalent to the following pattern:
 $$Q(root, \alpha) \gets R_f(root, \alpha, x), R_g(x, \alpha).$$
Later subsections will detail how this conversion is done,
 but note how the auxiliary variable $x$ captures the structural
 constraint from the pattern.
Evaluating $Q$ with a simple hash join strategy exemplifies the
 benefits of the relational approach:
 it considers structural \emph{and} equality constraints
 (in this case by doing a hash join keyed on $(x, \alpha)$);
 indeed, the relational perspective sees
 no difference between the two kinds of constraints.

\begin{figure}
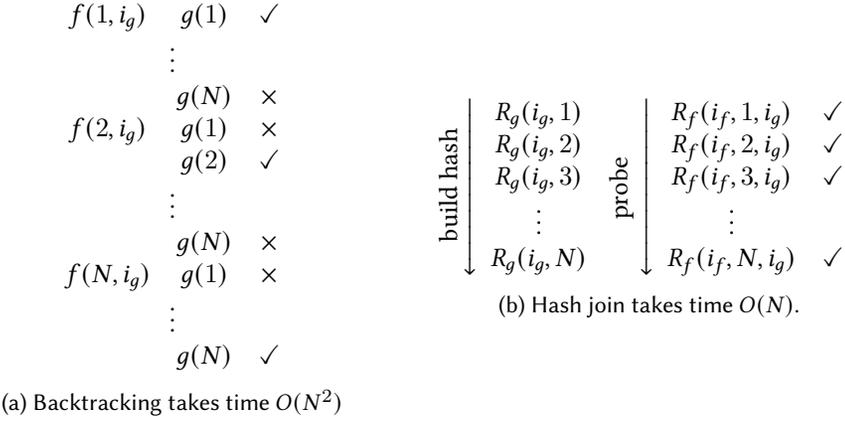

  \begin{subfigure}{0.45\textwidth}
    \centering
    \begin{tabular}{ c c c }
      $f(1, i_{g})$ & $g(1)$ & \checkmark\\
      \multicolumn{3}{c}{\vdots} \\
                    & $g(N)$ & $\times$\\
      $f(2, i_{g})$ & $g(1)$ & $\times$\\
                    & $g(2)$ & \checkmark\\
      \multicolumn{3}{c}{\vdots} \\
                    & $g(N)$ & $\times$\\
      $f(N, i_{g})$ & $g(1)$ & $\times$\\
      \multicolumn{3}{c}{\vdots} \\
                    & $g(N)$ & \checkmark\\
    \end{tabular}
    \caption{ Backtracking takes time $O(N^{2})$ }\label{tab:hj}
  \end{subfigure}
  \begin{subfigure}{0.45\textwidth}
    \centering $\rotatebox[origin=c]{90}{build hash}%
    \left\downarrow
      \begin{tabular}{ c }
        $R_g(i_g, 1)$ \\
        $R_g(i_g, 2)$ \\
        $R_g(i_g, 3)$ \\
        $\vdots$ \\
        $R_g(i_g, N)$ \\
      \end{tabular}
    \right.$ $\rotatebox[origin=c]{90}{probe}%
    \left\downarrow
      \begin{tabular}{ c c }
        $R_f(i_f, 1, i_g)$ & \checkmark \\
        $R_f(i_f, 2, i_g)$ & \checkmark\\
        $R_f(i_f, 3, i_g)$ & \checkmark\\
        $\vdots$ & \\
        $R_f(i_f, N, i_g)$ & \checkmark\\
      \end{tabular}
    \right.$
    \caption{ Hash join takes time $O(N)$. }\label{tab:hj}
  \end{subfigure}
  \caption{
    \Ematching $f(\alpha, g(\alpha))$
    with backtracking search and a simple hash join
    on the \egraph/database in \autoref{fig:egraph-and-table}.
  }\label{fig:emvhj}
\end{figure}

This observation leads us to a very simple
 algorithm for relational \ematching,
 shown in \autoref{alg:main}.
Relational \ematching~takes an e-graph $E$ and a set
 of patterns $\textit{ps}$.
It first transforms the \egraph to a relational database $I$.
Then, it reduces every pattern $p$ to a conjunctive query $q$.
Finally, it evaluates the conjunctive queries over $I$.
These intermediate steps will be detailed in the following subsections.

\begin{algorithm}
  \LinesNumbered{} \SetAlgoLined{}
  \KwIn{An \egraph~$E$ and a list of \ematching~patterns $\textit{ps}$}
  \KwOut{The result of running $\textit{ps}$ on $E$}
  $I  \gets \textsc{EGraphToDatabase}(E)$\;
  $qs \gets \{ \textsc{PatternToCQ}(p) \mid p \in ps \}$\;
  \Return $\{\textsc{EvalCQ}(q, I)\ \mid q \in qs \}$
  \caption{\textsc{RelationalEMatching}}
  \label{alg:main}
\end{algorithm}

\subsection{From the E-Graph to a Relational Database}

% \begin{figure}
%   \centering
%   \begin{align*}
%     \textsc{EgraphToDatabase}(G) &= \{\textsc{EnodeToAtom}(n)\ \mid\text{ for \enode\ } n\in G\}\\
%     \textsc{EnodeToAtom}(f(c_1,\ldots,c_k)) &= R_f(\textit{lookup}(f(c_1,\ldots,c_k)), c_1,\ldots, c_k)
%   \end{align*}
%   \caption{Transforming an \egraph~to relational database}
%   \label{fig:egraph_to_db}
% \end{figure}

\begin{figure}
  \begin{subfigure}{0.45\textwidth}
    \centering \includegraphics[width=0.65\linewidth]{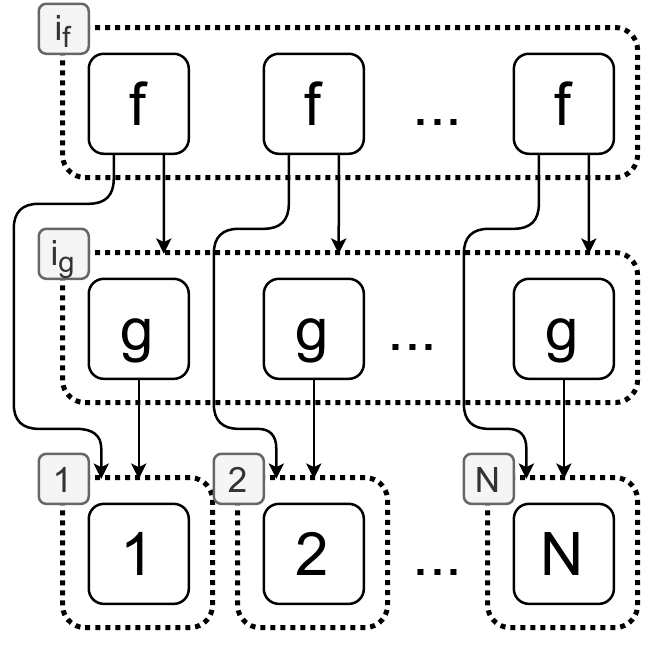}
    \caption{An example \egraph, reproduced from \autoref{fig:egraph-eg}.}\label{fig:egraph-eg2}
  \end{subfigure}
  \hfill
  \begin{subfigure}{0.25\textwidth}
    \centering
    \begin{tabular}{ c c c }
      \toprule
      id & $\text{arg}_{1}$ & $\text{arg}_{2}$ \\
      \midrule
      $i_{f}$ & 1 & $i_{g}$ \\
      $i_{f}$ & 2 & $i_{g}$ \\
      \vdots & \vdots & \vdots \\
      $i_{f}$ & $N$ & $i_{g}$ \\
      \bottomrule
    \end{tabular}
    \caption{Relation of $f$.}\label{tab:rel-f}
  \end{subfigure}
  \hfill
  \begin{subfigure}{0.25\textwidth}
    \centering
    \begin{tabular}{ c c }
      \toprule
      id & $\text{arg}_{1}$ \\
      \midrule
      $i_{g}$ & 1 \\
      $i_{g}$ & 2 \\
      \vdots & \vdots \\
      $i_{g}$ & $N$ \\
      \bottomrule
    \end{tabular}
    \caption{Relation of $g$.}\label{tab:rel-g}
  \end{subfigure}
  \caption{ An \egraph~and its relational representation. Each \eclass~(dotted
    box) is labeled with its id.
  }\label{fig:egraph-and-table}
\end{figure}

The first step of relational e-matching is to transform the \egraph $E$
 into a relational database $I$.
The domain of the database is \eclass ids,
 and its schema is determined
 by the function symbols in $\Sigma$.
Every \enode with symbol $f$ in the \egraph
 corresponds to a tuple in the relation $R_f$ in the database.
If $f$ has arity $k$,
 then $R_f$ will have arity $k + 1$;
 its first attribute is the \eclass id that contains
 the corresponding \enode,
 and the remaining attributes are the $k$ children of the $f$ \enode.
\autoref{fig:egraph-and-table} shows an example \egraph
 and part of its corresponding database.
In particular, only the relations
 of function symbols $f$ and $g$ are presented in this figure.
There are $N$ other relations;
 each relation $R_j$ represents a constant
 $j$ and has exactly one tuple (i.e., singleton $(j)$).

We construct the database $I$
 by simply looping over every \enode in every \eclass in $E$
 and making
 a tuple in the corresponding relation:
\[
  I =
  \left\{
    R_f \gets (\find(i), \find(j_1), \ldots, \find(j_k))
    \mid
    M[i] = f(j_1, \ldots, j_k)
  \right\}
\]
Note that the tuples in the database contain only canonical \eclass
 ids returned from the \find{} function.

% An \egraph~$E$ representing a congruence relation over
% $T(\Sigma,\emptyset)$ will be mapped to a relational database over $C$, the set
% of \eclasses, with schema $S_C$, and every \enode~in $E$ will correspond to an
% atom in this relational database. Specifically, for every $k$-ary symbol $f$ in
% $\Sigma$, We add an $k+1$-ary symbol $R_f$ to $S_C$. We refer to fields of $R_f$
% as $\eclass, \text{arg}_1,\ldots,\text{arg}_k$ respectively, and they denote the
% \eclass~ an \enode~belongs to and the children \eclasses~of this
% \enode~respectively. For every \enode~in the \egraph, relational
% \ematching~insert an atom in the relational database. The algorithm is presented
% in \autoref{fig:egraph_to_db}.

% Take the \egraph~in Figure~\ref{fig:egraph-eg} for example. The transformed
% relation $R_f$ and $R_g$ for function symbol $f$ and $g$ are presented in
% Figure~\ref{tab:rel-f} and Figure~\ref{tab:rel-g}. The $f$-application
% \enodes~in \eclass~$c_f$ translate to atom $R_f(c_f, i,c_g)$ for $i\in[N]$,
% because the \enodes~are contained in \eclass~$c_f$ and have children $i$ and
% $c_g$ for some $i\in[N]$. Similarly, the $g$-application \enode~in \eclass~$c_g$
% translates to the atom $R_g(c_g, i)$ for $i\in[N]$. \autoref{fig:egraph_to_db}
% only shows the relation representing $f$ and $g$. Besides the relation $R_f$ and
% $R_g$ presented, there are $N$ other relations, each representing a constant
% symbol and contains one entry that denotes the \eclass~where the corresponding
% \enode~is contained.

Our presentation in this paper specifically targets
 \ematching use cases like equality saturation,
 where the building of the database $I$ can be amortized.
In this setting, \ematching is done in large batches,
 and expensive work like congruence closure can be
 amortized between these batches using a technique called ``rebuilding''~\cite{egg}.
The time complexity of building this database is always linear,
 which is subsumed by the time complexity
 of most non-trivial \ematching patterns.
In \autoref{sec:incremental},
 we discuss how this technique could be
 generalized to the non-amortized setting of
 frequently updated \egraph as future work.

\subsection{From Patterns to Conjunctive Queries}
\label{sec:pattern-to-query}

\def\aux{\textsc{Aux}} \def\compile{\textsc{Compile}}
\begin{figure}
  \begin{align*}
    \compile(p) &= Q(\textit{root}, v_1,\ldots,v_k) \leftarrow \textit{atoms}\\
                &\text{\quad where $v_1\dots v_k$ are variables in $p$}\\
                &\text{\quad and $\aux(p) = \textit{root} \sim \textit{atoms}$}
    \\
    \aux(f(p_1,\ldots, p_k)) &= v \sim R_f(v, v_1, \ldots, v_k), A_1, \ldots, A_k \\
                & \text{\quad { where} $v$ is fresh and }\aux(p_i) = v_i \sim A_i
    \\
    \aux(x) &= x \sim \emptyset \text{\qquad { where} $x$ is a pattern variable}\\
  \end{align*}
  \caption{Compiling a pattern to a conjunctive query.}\label{alg:compile}
\end{figure}

Once we have a database that corresponds to the \egraph,
 we must convert each pattern we wish to \ematch to a conjunctive query.
We use the algorithm in
\autoref{alg:compile} to ``unnest'' a pattern to a conjunctive query
 by connecting nested patterns with \textit{auxiliary variables}.

The \aux{}
function returns a variable and a conjunctive query atom list.
Particularly, for
non-variable pattern $f(p_1,\ldots p_k)$,
\aux{} produces a fresh variable $v$
and a concatenation of $R_f(v, v_1,\ldots, v_k)$ and atoms from $A_i$, where
$v_i\sim A_i$ is the result of calling $\aux(p_i)$.
For variable pattern $x$,
 \aux{} simply returns $x$ and an empty list.
Note that the auxiliary variables introduced by
 $\aux(f(\ldots))$ are not included in the head of the query,
 and thus are not part of the output.

Given a pattern $p$, the \compile{}
function returns a conjunctive query with body atoms from $\aux(p)$ and the head
atom consisting of the root variable and variables in $p$. The compiled
conjunctive query and the original \ematching~query are equivalent because there
is a one-to-one correspondence between the output of them.
Specifically, each
 \ematching output
 $(i_{\textit{root}}, \sigma)$
 corresponds to a query output of
 $\{ \textit{root} \mapsto i_{\textit{root}} \} \cup \sigma$.
The only difference is that
 returning the root \eclass id is a special consideration for \ematching,
 but it is just another variable in the conjunctive query.

The \compile{} function (specifically the \aux{} subroutine) relies on
 the fact that the database contains only canonical \eclass ids.
Without this fact, nested patterns would require an additional join
 on the equivalence relation $\equivid$.
But since $i \equivid j \iff i = j$
 if $i$ and $j$ are canonical \eclass ids,
 we can omit introducing the additional join,
 instead joining nested patterns directly on the auxiliary variable.

Using this algorithm, the example pattern $f(\alpha, g(\alpha))$ is compiled to
the following conjunctive query:
\begin{equation}
  \label{eqn:nonlinear-cq}
  Q(\mathit{root},\alpha)\gets
  R_f(\mathit{root},\alpha,x),R_g(x,\alpha).
\end{equation}
Compared to the original
\ematching~pattern, this flattened representation enables relational \ematching to
utilize both the structural and the equality constraints. For example, a reasonable
query plan that database optimizers will synthesize is a hash join
on both join variables (i.e., $x$ and $\alpha$), which takes
$O(N)$ time. In contrast, backtracking-based \ematching{} takes $O(N^2)$ time.

Figure~\ref{fig:emvhj} shows the traces for running a direct backtracking search on the
\egraph~and running hash join on the relational representation.
Every term enumerated by hash join will simultaneously satisfy all the constraints.
Conceptually, backtracking-based \ematching~can be seen as a hash join that only builds and
look-ups a single variable (i.e., $x$),
 and filters the outputs using the equality predicate on $\alpha$.
In other words, existing \ematching~algorithms
 will consider all $f(\alpha, g(\beta))$ terms regardless of whether $\alpha$ is congruent
 to $\beta$,
 while the generated conjunctive query gives the query optimizer the
 freedom to synthesize query plans that will consider only tuples where $\alpha\cong\beta$.

\subsection{Answering CQs with Generic Join}

% \begin{figure}
%   \includegraphics[width=0.9\linewidth]{trie.pdf}
%   \begin{minipage}{0.45\textwidth}
%     \subcaption{Trie for relation $R_f$ (\autoref{tab:rel-f}).}
%     \label{fig:trie-Rf}
%   \end{minipage}
%   \begin{minipage}{0.45\textwidth}
%     \subcaption{Trie for relation $R_g$ (\autoref{tab:rel-g}).}
%     \label{fig:trie-Rg}
%   \end{minipage}
%   \caption{
%     The tries built
%     for pattern $f(\alpha, g(\alpha))$
%     using variable ordering $[\alpha, x, \textit{root}]$
%     on \egraph from \autoref{fig:egraph-and-table}.
%   }\label{fig:trie}
% \end{figure}

Finally, we consider the problem of efficiently solving the compiled conjunctive
queries. We propose to use the generic join algorithm to solve the generated
conjunctive queries. Although traditional query plans, which are based on two
way joins such as hash joins and merge-sort joins, are extensively used in
industrial relational database engine, they may suffer on certain queries
compiled from patterns. For example, consider the pattern
$f(g(\alpha),h(\alpha))$. The compiled conjunctive query is:
\begin{equation}
  \label{eqn:cyclic-cq}
  Q(\textit{root},\alpha)\gets
  R_f(\textit{root}, x, y), R_g(x, \alpha), R_h(y, \alpha).
\end{equation}
Like the classic triangle query,
 this is a cyclic conjunctive query (\autoref{sec:hypergraph}).
We call \ematching~patterns that generate cyclic conjunctive
 queries \textit{cyclic patterns}.
For such cyclic queries,
 \citet{wcoj} show there exist databases on which {\em any}
 two-way join plan is suboptimal.
In contrast, generic join is guaranteed to run in time linear
 to the worst case output size.
Moreover, generic join can have comparable performance on acyclic
 queries with two-way join plans.
These properties make generic join our ideal
 solver for conjunctive queries generated from \ematching patterns.

% Suppose we fix the variable ordering $\pi$ to be $[x,y,\alpha,root]$.
% We first build a trie for $R_f$ on index
% $(\text{arg}_1,\text{arg}_2,\eclass)$ and a trie for $R_g$ on index
% $(\eclass, \text{arg}_1)$ as a preprocessing before actually  running the
% generic join
% algorithm. Then, the algorithm will intersect on
% column $\text{arg}_1$ of $R_f$ and column $\eclass$ of $R_g$, which nails down
% the possible values of $x$ and obtains residual relations $R_f[\text{arg}_1=x]$ and
% $R_g[\eclass=x]$. Next, it will intersect on column $\text{arg}_2$ of $R_f$ and
% column $\eclass$ of $R_g$ to nail down the values of $y$ and obtain residual
% relations $R_f[\text{arg}_1=x,\text{arg}_2=y]$ and $R_g[\eclass=y]$. To obtain a
% value for $\alpha$, the algorithm will then intersect on $R_g[\eclass=x]$ and
% $R_g[\eclass=y]$, and it will finally enumerate
% $R_f[\text{arg}_1=x,\text{arg}_2=y]$ to decide values for $\textit{root}$.

Using the generic join algorithm, suppose we fix the variable ordering to be
$[\alpha, x, \textit{root}]$ on the generated conjunctive query \ref{eqn:nonlinear-cq}.
The algorithm below shows generic join instantiated on this particular CQ:
% We first build
%  a trie for $R_f$ on index $(\mathit{child}_1, \mathit{child}_2, \mathit{root})$
%  a trie for $R_g$ on index $(\mathit{child}_1, \mathit{root})$
%  (tries shown in \autoref{fig:trie-Rf} and Figure~\ref{fig:trie-Rg}).
% Then, the algorithm will intersect on column $\mathit{child}_1$ of $R_f$ and on
% column $\mathit{child}_1$ of $R_g$, which nails down the possible values of
% $\alpha$ and obtains residual relation $R_f[\mathit{child}_1=\alpha]$ and
% $R_g[\mathit{child}_1=\alpha]$. Next, it will intersect on column $\mathit{child}_2$
% of $R_f[\mathit{child}_1=\alpha]$ and on $\eclass$ of
% $R_g[\mathit{child}_1=\alpha]$ to obtain possible values of $x$ and residual
% relation $R_f[\mathit{child}_1=\alpha, \mathit{child}_2=x]$. Finally, we enumerate
% terms of $R_f[\mathit{child}_1=\alpha, \mathit{child}_2=x]$ to find possible values
% of $\mathit{root}$.

\begin{algorithm}
  \LinesNumbered{} \SetAlgoLined{}
  \KwResult{compute $Q(\mathit{root},\alpha) \gets R_f(\mathit{root},\alpha,x),R_g(x,\alpha)$}
  \tcp{compute all possible values of $\alpha$}
  $A = R_f(\mathit{root},\alpha,x).\alpha \cap R_g(x,\alpha).\alpha$\;
  \For(){$i_\alpha \in A$} {
    \tcp{compute all possible values of $x$ given $\alpha=i_\alpha$}
    $X = R_f(\mathit{root},i_\alpha,x).x \cap R_g(x,i_\alpha).x$\;
    \For(){$i_x \in X$} {
      \tcp{compute all possible values of $\mathit{root}$ given $\alpha=i_\alpha$ and $x=i_x$}
      $\mathit{Roots} = R_f(\mathit{root},i_\alpha,i_x).\mathit{root} \cap R_g(i_x,i_\alpha).\mathit{root}$\;
      \For(){$i_{\mathit{root}} \in \mathit{Roots}$}{
        output$(i_{\mathit{root}}, i_\alpha)$ } } }
  \caption{
    Relational \ematching using GJ for $f(\alpha, g(\alpha))$,
    with ordering $[\alpha, x, \mathit{root}]$.
  }\label{alg:gj-pattern}
\end{algorithm}

\subsection{Complexity of Relational \Ematching}
% The translation of \ematching to query-answering on a database
%  is only productive if the size of the database is comparable to
%  that of the \egraph.
% Our relational encoding satisfies this requirement:
% \begin{lemma}\label{lem:bounded}
%   Let $I$ be the database representing the \egraph,
%    and let $N, C$ be the number of \enodes and \eclasses
%    in the \egraph, respectively.
%   We then have $\sum_{R \in I} |R| \leq N + C$,
%    where $|R| = \sum_{t \in R} |t|$
%    is the total size of the tuples in $R$,
%    and each value in the tuple has constant size.
% \end{lemma}

% \begin{proof}
%   Every \enode  corresponds to exactly one tuple,
%    and every tuple has exactly one extra attribute
%    than its \enode has children.
%   In addition, there may be \eclasses that are not children of
%    any \enodes.
%   Therefore $I$ has exactly $N$ tuples
%    and at most $N+C$ values.
% \end{proof}

Generic join guarantees worst-case optimality with respect to the output size,
 and relational \ematching preserves this optimality.
In particular, we have the following theorem:
\begin{theorem}
  Relational \ematching is worst-case optimal; that is, fix a pattern $p$, let
  $M(p,E)$ be the set of substitutions yielded by \ematching on an \egraph $E$
  with $N$ \enodes, relational \ematching runs in time $O(\max_E(|M(p,E)|))$.
\end{theorem}

\begin{proof}
  Notice that there is an one-to-one correspondence between output tuples of
  the generated conjunctive query and the \ematching pattern. Therefore, the
  worst-case bound is the same across an \ematching pattern and the conjunctive
  query it generated. Because generic join is worst-case optimal, relational
  \ematching also runs in worst-case optimal time with respect to the output
  size.
\end{proof}

The structure of \ematching patterns allows us to derive an additional bound
 dependent on the {\em actual} output size rather than the worst-case output size.
% Let $|\mathit{out}|$ denote the number of matches found,
%  we have the following:
\begin{theorem}
  Fix an \egraph $E$ with $N$ \enodes that compiles to a database $I$,
   and a fix pattern $p$ that compiles to conjunctive query
  $Q(\overline{X}) \gets R_1(\overline{X_1}),\ldots,R_m(\overline{X_m})$.
  Relational \ematching $p$ on $E$
   runs in time
   $O\left(\sqrt{|Q(I)| \times \Pi_{i}|R_i|}\right)
   \leq
    O\left(\sqrt{|Q(I)| \times N^m}\right)
   $.
% of a pattern with $|A|$ \enodes\
%    on an \egraph\ of size $N$ runs in time $O(\sqrt{N^m|Q(I)|})$.
\end{theorem}

\begin{proof}
  Let $\overline{X^\circ}$ be the set of isolated variables,
   those that occur in only one atom.
  Note that $\overline{X^\circ} \subseteq \overline{X}$,
   since $\overline{X}$ is precisely the pattern variables and the root,
   and auxiliary variables must occur in at least two atoms.
  Using these,
   define two new queries:
  \begin{align*}
  C(\overline{X^\circ}) & \gets R_{1}(\overline{X_{1}}), \ldots, R_m(\overline{X_m}) \\
  Q'(\overline{X}) & \gets R_{1}(\overline{X_{1}}), \ldots, R_m(\overline{X_m}), C(\overline{X^\circ})
  \end{align*}
 Since $\overline{X^\circ} \subseteq \overline{X}$,
  $C$ is the same query as $Q$ but with zero or more variables projected out.
 Therefore,
  every tuple in $C(I)$ corresponds to one in the output $Q(I)$,
  so $C(I) \subseteq Q(I)$ and $|C(I)| \leq |Q(I)|$.

 Now we can compute the AGM bound for $Q'$.
 Our new atom $C(\overline{X^\circ})$
  includes all those variables that only appear in one atom of $Q$.
 Therefore, every variable in $Q'$ occurs in at least two atoms,
  so assigning $1/2$ to each edge is a fractional edge cover.
 Thus:
 \begin{align*}
   \textsf{AGM}(Q') &= \sqrt{|C(I)| \times \Pi_i |R_i|} \\
                    &\leq \sqrt{|Q(I)| \times \Pi_i |R_i|} &\text{\quad since $|C(I)| \leq |Q(I)|$} \\
                    &\leq \sqrt{|Q(I)| \times N^m} &\text{\quad since $|R_i| < N$}
 \end{align*}

 Let $\GJ(Q', I)$ denote the running time
  of generic join with query $Q'$ on database $I$.
 We know that $\GJ(Q', I) \leq \textsf{AGM}(Q')$.
 Because $C(I) \subseteq Q(I)$, we also know $Q'(I) = Q(I)$,
  and we can use $GJ(Q', I)$ to bound $GJ(Q, I)$.
 Now we show that $\GJ(Q, I) \leq \GJ(Q', I)$.

 The query $Q'$ is just $Q$ with an additional
  atom $C(\overline{X^\circ})$ that
  covers the variables that only appeared in one atom from $Q$.
 Fix a variable ordering for generic join that puts those variables
  in $\overline{X^\circ}$ at the end.
 So loops of both \GJ instantiations are the same,
  except that, in $Q'$,
  each loop corresponding to a variable in $\overline{X^\circ}$ performs an intersection with $C$,
  but not in $Q$.
 But these intersections are in the innermost loops,
  at which point all intersections with atoms from $Q$ have already been done.
 So the intersections with $C$ do nothing,
  since $C$ is precisely $Q$ projected down to the variables in $\overline{X^\circ}$!
 Since those intersections are not helpful and $Q$ simply does not do them,
  $\GJ(Q, I) \leq \GJ(Q', I)$.

 Putting the inequalities together,
  we get:
 \[
   \GJ(Q, I)
   \leq
   \GJ(Q', I)
   \leq
   AGM(Q')
   \leq
   O\left(\sqrt{|Q(I)| \times \Pi_{i}|R_i|}\right)
   \leq
    O\left(\sqrt{|Q(I)| \times N^m}\right)
 \].
\end{proof}

\begin{example}[Complexity of relational \ematching]
  Consider the pattern $f(g(\alpha))$,
   which compiles to the query $Q(r, \alpha) \gets R_{f}(r, x), R_{g}(x, \alpha)$.
  Following the proof we define $C(r,\alpha) \gets R_{f}(r,x), R_{g}(x,\alpha)$
   and
   $Q'(r, \alpha) \gets R_{f}(r, x), R_{g}(x, \alpha), C(r, \alpha)$.
  The AGM bound for $Q'$ is $N^{1/2}N^{1/2}|C|^{1/2} = N\sqrt{|C|} = N\sqrt{|Q|}$.
  This also bounds the run time of generic join on $Q$.
\end{example}

The above bound is tight for linear patterns,
 in which case each variable occurs exactly twice in $Q'$.
In the case of nonlinear patterns,
 we may find tighter covers than assigning $1/2$ to each atom,
 thereby improving the bound.

\subsection{Supporting Multi-patterns}

Multi-patterns are an extension to \ematching used in both
 SMT solvers~\cite{efficient-ematching}
 and program optimizations~\cite{tensat}.
A multi-pattern is a list of patterns of the form $(p_1, \ldots, p_k)$
 that are to be simultaneously matched (i.e., the instantiation of each contained pattern
should use the same substitution $\sigma$).
For example, \ematching the multi-pattern
 $(f(\alpha, \beta), f(\alpha, \gamma))$ searches for pairs of two $f$-applications
 whose first arguments are equivalent.
Efficient support for multi-patterns
 on top of backtracking search requires complicated additions
 to state-of-the-art \ematching algorithms~\cite{efficient-ematching}.
Relational e-matching supports multi-patterns
 ``for free'':
 a multi-pattern is compiled to a single
 conjunctive query just like a single pattern.
For example, the conjunctive query for the
 multi-pattern $(f(\alpha, \beta), f(\alpha, \gamma))$ is
\begin{equation}
  Q(\textit{root}_1,\textit{root}_2,\alpha,\beta, \gamma)\gets
  R_f(\textit{root}_1,\alpha,\beta), R_g(\textit{root}_2, \alpha, \gamma).
\end{equation}
This is one example that shows the wide applicability of the relational
 model adopted in relational \ematching.

%%% Local Variables:
%%% mode: latex
%%% TeX-master: "main"
%%% End:

\section{Optimizations}
\label{sec:optimization}

Our implementation of relational \ematching
 using generic join is simple (under 500 lines),
 but that does not preclude having several optimizations
 important for practical performance.

\subsection{Degenerate Patterns}
\label{sec:degenerate}

Not all patterns correspond to conjunctive queries that involve relational
joins.
Non-nested patterns (whether linear or non-linear) will produce
 relational queries without any joins:
\begin{align*}
  f(\alpha, \beta) &\leftrightarrow R_f(\mathit{root}, \alpha, \beta) \\
  f(\alpha, \alpha) &\leftrightarrow R_f(\mathit{root}, \alpha, \alpha)
\end{align*}
The corresponding query plan is simply a
 scan of a relation with possible filtering.
For these queries,
 generic join (or any other join plan) offers no benefit,
 and building the indices for generic join incurs unnecessary overhead.
A relational \ematching implementation (or any other kind)
 should have a ``fast path''
 for these relatively common types of queries
 that simply scans the \egraph/database for matching \enodes/tuples.
For this reason,
 we exclude these kinds of patterns from our evaluation in \autoref{sec:eval}.

\subsection{Variable Ordering}

Different variable orderings can dramatically affect performance
 for generic join~\cite{eval-wcoj, emptyheaded},
 so choosing a variable ordering is important.
Compared to join plans for binary joins,
 query plans for generic join is much less studied.
In relational \ematching\ we choose a variable ordering
 using two simple heuristics:
First, we prioritize variables that occur in many relations,
 because the intersected set of many relations is likely to be smaller.
Second, we prioritize variables that occur in small relations,
 because intersections involving a small relations are also likely to be smaller.
Performing smaller intersections first can prune down the search space early.

Using these two heuristics,
 the optimizer is able to find more efficient query plans
 than the top-down search of backtracking-based \ematching.
This is even true for linear patterns,
 where our relational \ematching\ has no more information than
 \ematching, but it does have more flexibility.
Consider the linear pattern $f(g(h(\alpha)))$ compiled to
the query $R_f(\mathit{root}, x), R_g(x, y), R_h(y, \alpha)$.
When there are very few $h$-application \enodes\ in the \egraph,
 $R_h$ will be small.
The variable ordering $[y, x, \mathit{root}, \alpha]$
 takes advantage of this by
 intersecting $R_g.y \cap R_h.y$ first,
 resulting in an intersection no larger than $R_h$.
This ``bottom-up'' traversal is not possible in
 conventional \ematching.

%  For example, for queries like $+(a, \alpha)$, our optimizer is able to
% exploit the fact that there is only one atom in $R_a$ (i.e., the relation
% representing constant $a$), so enumerating this relation and use attributes of
% $R_a$ to prune $R_+$ will immediately produces all valid substitutions (assuming
% the indices on $R_+$ is built).
% Meanwhile, in the top-down fashion of backtracking-based \ematching, the
% matching procedure will need to go through all atoms of $R_+$ and check if
% constant $a$ resides in their second child. Such query plan examines much more
% unnecessary atoms in $R_+$, because only a small number of $+$-application
% \enodes{} are connected to $a$.

\subsection{Functional Dependencies}
\label{sec:fd}

Functional dependencies describe the dependencies between columns. For example,
a functional dependency on relation $R(y, x_1, x_2, x_3)$ of the form
$x_1,x_2,x_3\rightarrow y$ indicates that
for each tuple of $R$,  the values of $x_1$, $ x_2$, and $x_3$ uniquely
determines the value $y$. Functional dependencies are ubiquitous in relational
\ematching. In fact, every transformed schema of the form $R_f(\textit{e-class},
\textit{arg}_1,\ldots, \textit{arg}_k)$ has a functional dependency from
$\textit{arg}_1,\ldots,\textit{arg}_k$ to $\textit{e-class}$.
When the variable graph formed by functional dependencies is
 acyclic~\footnote{Cyclicity of functional dependencies is unrelated to cyclicity of the query.},
 we can speed up generic join by ordering the variables to follow the
 topological order of the functional dependency graph.
Every conjunctive query compiled from an \ematching\ pattern
 has acyclic functional dependencies,
 because each dependency goes from the \enode's children to the \enode's parent \eclass.
Relational \ematching~can therefore always choose a variable
 ordering that is consistent with the functional dependency.
Our implementation tries to respect functional dependencies,
 but prioritizes larger intersections more.

As an example, consider conjunctive query~\ref{eqn:cyclic-cq} again. It is
synthesized from pattern $f(g(\alpha), h(\alpha))$ and, assuming each relation
has size $N$, an AGM bound of $O(N^{3/2})$. Suppose however that we pick the
variable ordering $\pi$ to be $[\alpha, x, y, \textit{root}]$. For every
possible value of $\alpha$ chosen, there will be at most one possible value for
$x,y,$ and $\textit{root}$ by functional dependency, which can be immediately
determined. This reduces the run time from $O(N^{3/2})$ to $O(N)$.

\subsection{Batching}
Generic join always processes one variable at a time,
 even if multiple consecutive variables are from the same atom.
We find this strategy to be inefficient in practice,
 as it results in deeper recursion that does little useful work.

Consider the query $Q(x, y, z, w) \gets R(x, y, \alpha), S(\alpha, z, w)$.
The right variable ordering places $\alpha$ at the front,
 since it is the only intersection.
We observe that variables that only appear
 in one atom can be ``batched''
 with others that only appear in the same atom.
Batched variables are treated as a single variable
 in the trie and intersections.
So instead of variable ordering $[\alpha, x, y, z, w]$,
 we can use $[\alpha, (x, y), (z, w)]$.
This lowers the recursion depth of generic join (from 5 to 3)
 and improves data locality by reducing pointer-chasing.

% For example, suppose the algorithm need to loop over a 2-level
%  trie with two nested loops (without intersection),
%  where the outer loop is large but the inner loop is always small.
% Each loop incurs overhead due to caching,
%  processing the trie with nested loops is slower than
%  simply scanning a flattened vector.
% For this reason, we flatten the last levels of each trie where
%  the variables do not participate in any intersection,
%  thereby ``batching'' the variables.

%%% Local Variables:
%%% mode: latex
%%% TeX-master: "main"
%%% End:

\section{Evaluation}
\label{sec:eval}
\label{sec:evaluation}

To empirically evaluate relational \ematching,
 we implemented it inside the
 \egg equality saturation toolkit~\cite{egg}.
Our implementation consists
 about 80 lines of Rust inside \egg itself to convert
 patterns into conjunctive queries,
 paired with a
 a separate, e-graph-agnostic
 Rust library to implement generic join in fewer than $500$ lines.

\egg's existing
 \ematching infrastructure is also about $500$ lines of Rust,
 and it is interconnected to various other parts of \egg.
Qualitatively,
 we claim that the relational approach is simpler to implement,
 especially since the CQ solver is completely modular.
 We could plug in a different generic join implementation~\footnote{
   There is no reusable generic join implementation at the time of writing.},
 or even a more conventional binary join implementation.

In this section, we refer to \egg's existing \ematching implementation as ``\EM''
 and our relation approach as ``\GJ.''

\subsection{Benchmarking Setup}

We use \egg's two largest benchmark suites
 as the basis for our two benchmark suites.
The \texttt{math} suite implements a
 simple computer algebra system,
 including limited support for symbolic differentiation and integration.
The \texttt{lambda} suite implements a
 partial evaluator for the lambda calculus.
Each \egg benchmark suite provides a set of rewrite rules,
 each with a left and right side pattern,
 and a set of starting terms.

To construct the \egraphs used in our benchmarks
 we ran equality saturation
 on a set of terms selected from \egg's test suite,
 stopping before the \egraph reached
 1e5, 1e6, 2e6, and 3e6 \enodes.
The result is four increasingly large \egraphs
 for each benchmark suite
 filled with terms that are
 generated by the suite's rewrite rules.
For each benchmark suite
 and each of the four \egraph sizes,
 we then ran \ematching on the \egraph
 using both \EM and \GJ.
We ran each approach 10 times and
 took the minimum run time.

For our \GJ approach,
 we ran each trial twice.
The first time builds the index tries
 necessary for generic join just-in-time,
 and the run time includes that.
On the second trial,
 \GJ uses the pre-built index tries from the first run,
 so the time to build them is excluded.
In both \autoref{fig:speedup} and \autoref{tab:eval},
 orange bars/rows show the first runs (including indexing),
 and blue bars/rows show the second runs (excluding indexing).

All benchmarks are single-threaded,
 and they were executed on a
 4.6GHz CPU with 32GB of memory.

\begin{figure}
  \centering
  \includegraphics[width=1.0\linewidth]{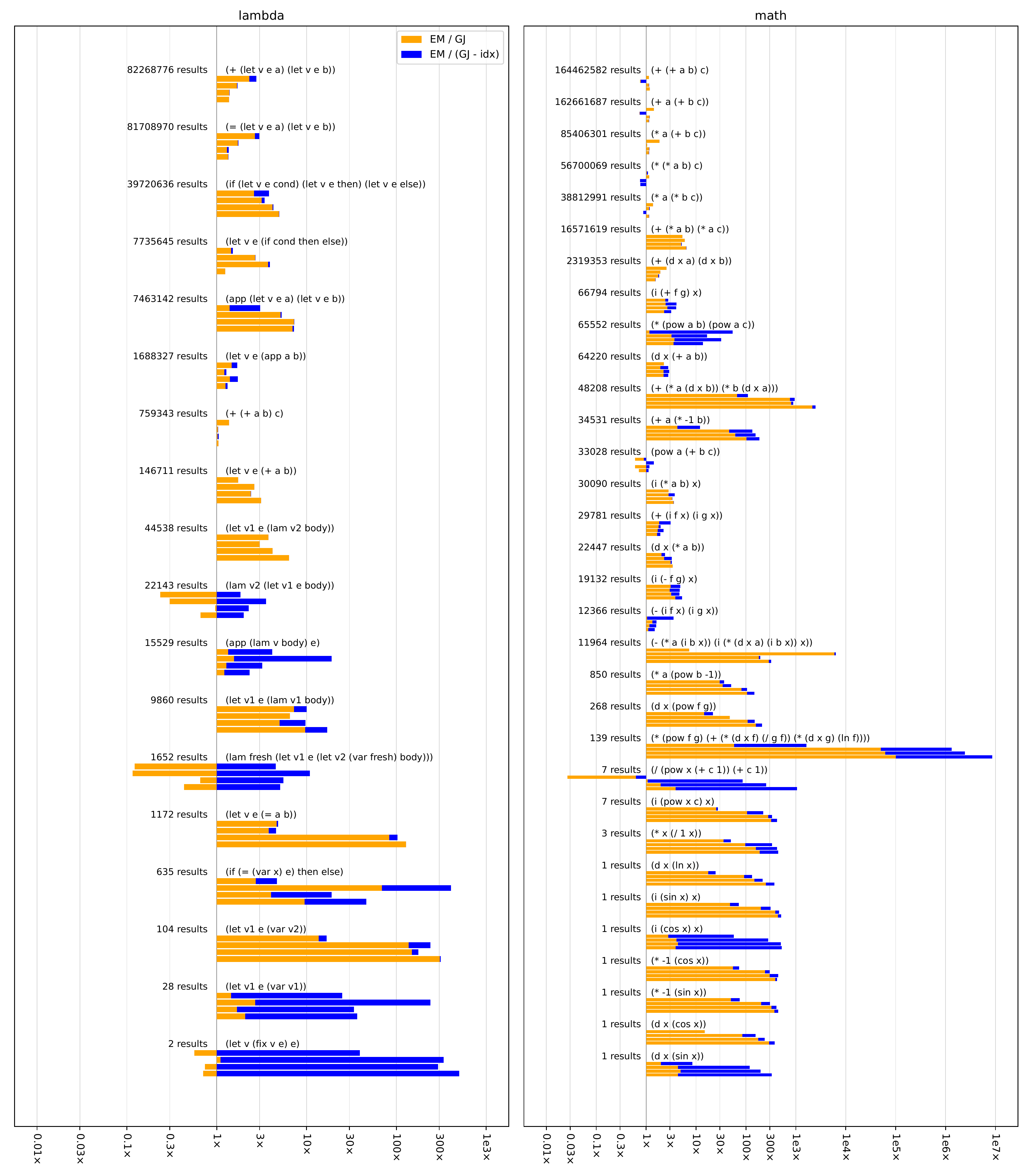}
  \caption{
    Relational e-matching can be up to 6 orders of magnitude faster
     than traditional e-matching on complex patterns.
    Speedup tends to be greater when the output size is smaller.
    Bars to the right of the ``$1\times$'' line
     indicate that relational \ematching is faster.
    The plots show two benchmark suites,
     \texttt{lambda} and \texttt{math},
     taken from the \egg test suite.
    Each group of bars
     shows the benchmarking results of e-matching a single
     pattern on 4 increasingly large e-graphs (top to bottom),
     comparing
     \egg's built-in e-matching (\EM)
     with our relational e-matching approach using generic join (\GJ).
    % Patterns are sorted
    %  by the number of \ematching results on the largest \egraph
    %  (shown to the right of the pattern).
    The orange bar shows the multiplicative speedup of our approach:
     $\EM / \GJ$.
    The blue bar shows the same, but \textit{excluding}
     the time spent building the indices needed for generic join:
     $\EM / (\GJ - \mathsf{idx})$.
    % Red ``$<$'' markers to the right of bars indicate that \EM timed out after 60 seconds.
    The text above each group of bars shows
     the pattern itself
     and
     the number of substitutions found on the largest \egraph;
     the patterns are sorted by this quantity.
  }\label{fig:speedup}
\end{figure}

\begin{table}
  \centering
  \caption{
    Summary statistics across patterns for each benchmark suite.
    The ``Idx'' column shows whether the time to build indices in \GJ
     is included (+) or not (--);
     the row color corresponds to the colors in \autoref{fig:speedup}.
    The ``Suite'' columns shows the benchmark suite,
     and the ``EG Size'' shows the number of \enodes
     in the \egraph used to benchmark.
    The ``\GJ'' and ``\EM'' columns show how many patterns
     that algorithm was fastest on in this configuration.
    % ``TO'' shows how many times \EM timed out
    %   (after 60 seconds); \GJ never timed out.
    ``Total'' shows the cumulative speedup across all patterns
      in that configuration.
    The remaining columns show statistics about the \EM/\GJ ratios
      for each pattern: harmonic mean, geometric mean, max, median, and min.
  }
  \rowcolors{2}{blue!20}{orange!20}
  \begin{tabular}{|llrrrrrrrrr|}
    \hline
    Idx  & Suite  & EG Size & \GJ & \EM & Total & HMean & GMean & Best  & Medn  & Worst   \\
    \hline
    +     & lambda & 4,142   & 15 & 3  & 1.69   & .84   & 1.71  & 13.62   & 1.60  & .12   \\
    --    & lambda & 4,142   & 18 & 0  & 2.58   & 2.99  & 4.23  & 39.17   & 3.68  & 1.10  \\
    +     & lambda & 57,454  & 16 & 2  & 2.60   & .95   & 2.66  & 136.54  & 2.65  & .12   \\
    --    & lambda & 57,454  & 18 & 0  & 2.87   & 3.33  & 9.11  & 406.70  & 4.05  & 1.03  \\
    +     & lambda & 109,493 & 15 & 3  & 1.66   & 1.75  & 3.11  & 148.96  & 2.03  & .65   \\
    --    & lambda & 109,493 & 18 & 0  & 1.70   & 3.32  & 7.46  & 291.18  & 4.10  & 1.05  \\
    +     & lambda & 213,345 & 15 & 3  & 2.20   & 1.55  & 3.40  & 304.33  & 1.72  & .43   \\
    --    & lambda & 213,345 & 18 & 0  & 2.21   & 2.96  & 8.23  & 501.12  & 5.04  & 1.04  \\
    \hline \rowcolor{white} \multicolumn{11}{c}{} \\
    \hline \rowcolor{white}
    Idx & Suite & EG Size & \GJ & \EM & Total  & HMean & GMean & Best         & Medn  & Worst \\
    \hline
    +   & math  & 8,205   & 30  & 2   & 5.49   & 0.64  & 4.61  & 66.54        & 2.79  & 0.03  \\
    --  & math  & 8,205   & 30  & 2   & 5.21   & 2.93  & 8.62  & 1,630.00     & 5.48  & 0.62  \\
    +   & math  & 53,286  & 29  & 3   & 311.23 & 2.61  & 13.50 & 50,030.29    & 3.62  & 0.74  \\
    --  & math  & 53,286  & 30  & 2   & 318.95 & 3.39  & 29.60 & 1,325,802.56 & 30.72 & 0.74  \\
    +   & math  & 132,080 & 29  & 3   & 96.55  & 2.66  & 15.18 & 61,488.73    & 4.02  & 0.60  \\
    --  & math  & 132,080 & 30  & 2   & 97.84  & 3.46  & 34.16 & 2,447,939.38 & 68.71 & 0.75  \\
    +   & math  & 217,396 & 30  & 2   & 119.82 & 2.83  & 18.34 & 101,023.37   & 3.91  & 0.72  \\
    --  & math  & 217,396 & 31  & 1   & 119.73 & 3.45  & 41.35 & 8,575,830.58 & 80.84 & 0.76  \\
    \hline
  \end{tabular}
  \vspace{1em}
  \label{tab:eval}
\end{table}

\subsection{Results}

\autoref{fig:speedup}
 show the results of our benchmarking experiments.
\GJ can be over 6 orders of magnitude faster
 than traditional e-matching on complex patterns.
Speedup tends to be greater when the output size is smaller,
 and when the pattern is larger and non-linear.
A large output indicates the \egraph\ is densely populated
 with terms matching the given pattern,
 therefore backtracking search wastes little time on unmatched terms,
 and using relational \ematching contributes little or no speedup.
Large and complex (non-linear) patterns require careful query planning
 to be processed efficiently.
For example, the pattern experiencing the largest speedup in~\autoref{fig:speedup}
 is 4 \enodes\ deep with 4 occurrences of the variable $f$.
Relational \ematching\ using generic join can devise a variable ordering
 to put smaller relations with fewer children on the outer loop,
 thereby pruning down a large search space early.
In contrast, backtracking search must traverse the \egraph\
 top-down.

In some cases index building time takes a significant proportion of the run time
 in relational \ematching, sometimes offsetting the gains.
Overall, relational \ematching\ remains competitive with the index building overhead.
In~\autoref{sec:incremental} we discuss potential remedies to alleviate this overhead.

\autoref{tab:eval}
 shows summary statistics across all patterns for each benchmark configuration.
Notably, \GJ is faster across patterns in
 every benchmarking configuration (the ``Total'' column).
 % the ``Total'' column shows that \GJ is faster overall than \EM
 % for every benchmarking configuration.
Much of the total benchmarking run time is dominated by
 simple, linear patterns (e.g. \texttt{(+ (+ a b) c)})
 that return many results.
Terms matching such patterns come to dominate
 the \egraph\ over time, due to the explosive expansion of
 associativity and commutativity.
As a result, the total speedup does not necessarily increase
 as the \egraph\ grows,
 whereas the best speedup as well as different average statistics
 steadily increase.

In summary, relational \ematching\ is almost always faster than backtracking search,
and is especially effective in speeding up complex patterns.

%%% Local Variables:
%%% mode: latex
%%% TeX-master: "main"
%%% End:

\section{Discussion}
\label{sec:discussion}

The relational model for \ematching is
 not only simple and fast,
 but it opens the door to a wide range of future work
 to further improve \egraphs and \ematching.

\subsection{Relational Representation of Code}

\update{
The representation of \enodes as relations and patterns as
 relational queries is simple and natural.
It may even feel obvious to someone familiar with Prolog-style programming,
 where a function with arity $k$ is often written as a relation with arity $k+1$, 
 with 1 extra argument for the output.
This relational representation can also be found in the literature on congruence closure.
For example, \citet{rummer2012matching} uses the same encoding to 
 simulate congruence closure procedures with a hyper-resolution calculus.
Unlike our focus on the performance of \ematching, 
 their goal is to improve the completeness of quantified reasoning in SMT solvers.
Research on large-scale code search and analysis~\cite{urma2013expressive, avgustinov2016ql, doop, Glean}
 has also explored storing program repositories in a database.
Programmers may issue queries against the database 
 to find code that match certain patterns, 
 which may be examples of API usage or 
 ``anti-patterns'' that pose security risks.
Certain complex queries can even express sophisticated analyses like 
 the points-to analysis~\cite{doop, avgustinov2016ql}. 
These code search and analysis engines need to balance the 
 need for expressiveness, efficiency, and ease of use.

While the idea to represent code as relations is not new, 
 we are the first, to the best of our knowledge, 
 to leverage relational join algorithms to speed up \ematching.
We designed optimizations specialized for \ematching, 
 and also derived the first non-trivial complexity bound
 for \ematching from a careful analysis of the generic join algorithm.
}
\subsection{Pushdown Optimization}

An \ematching pattern may have additional filtering
 condition associated with it.
For example, a rewrite rule with left hand side \verb|(* (/ x y) y)| may
 additionally require that $y\neq 0$.
When the variables involved in conditions all occur in a single relation (e.g.,
 $R_*$ and $R_/$),
 this relation can be effectively filtered
 even before being joined
 (e.g., using predicate $\find(\sigma(y))\neq\find(\lookup(0))$),
 which could immediately prune a large search space.

We call this \textit{pushdown optimization},
 which can be considered as
 \egraph's version of the relational query
 optimization that always pushes the filter operations down to the bottom
 of the join tree.
Note
 that the ability to do pushdown optimization stems
 from relational \ematching's ability to consider the constraints in any order;
 backtracking \ematching could not support this technique.
We currently do
 not implement this,
 because it requires breaking
 changes to \egg's interface.

Conditions that involve multiple variables can be ``pushed down'' as well.
In generic join,
 the filter can occur immediately after the variables appear in the variable ordering.
Thus,
 an implementation that supports these conditional filters
 should take this into account when determining variable ordering.

% The ``top-down'' approach of backtracking \ematching
%  cannot pre-filter
% To see this, consider the pattern $f(g(\alpha),g(\beta))$,
%  with the requirement $\sigma(\alpha)\neq 0$ and $\sigma(\beta)=0$.
% The \egraph cannot prune \enodes $g(i)$ using both
%  filtering condition $i=\find(\lookup(0))$
%  and $i\neq \find(\lookup(0))$ simultaneously.

\subsection{Join Algorithms}
Research in databases has proposed a myriad of different join algorithms.
For example, state-of-the-art database systems implement two-way joins
 like hash join and merge-sort join.
They have a longer history than generic join,
 and benefit from various constant factor optimizations.
Extensive research has focused on generating highly efficient query plans
 using two-way joins.
On the other hand, Yannakakis' algorithm~\cite{yannakakis}
 is proven to be optimal on a class of queries called {\em full ayclic queries},
 running in time linear to the total size of the input and the output.
All linear patterns correspond to acyclic queries,
 but some nonlinear patterns correspond to cyclic ones.
 Recent research~\cite{mhedhbi19,
 freitag} has also experimented with combining
 traditional join algorithms with generic join,
 achieving good performance.
In this paper we choose generic join for its simplicity,
 and future work may consider other join algorithms for relational \ematching.

\subsection{Incremental Processing}\label{sec:incremental}
We have focused on improving the core \ematching\ algorithm in this paper,
 yet prior work has successfully sped up \ematching\ by making it incremental~\cite{efficient-ematching}.
When the changes to the \egraph\ are small and the results of \ematching\
 patterns are frequently queried,
 maintaining the already-discovered matches becomes crucial for efficiency.
From our relational perspective, incremental \ematching\
 is captured precisely by the classic problem of
 incremental view maintenance (IVM)~\cite{DBLP:conf/vldb/CeriW91,
 DBLP:conf/sigmod/SalemBCL00,DBLP:conf/sigmod/ZhugeGHW95} in databases.
IVM aims to efficiently update an already-computed query result upon changes
 to the database, without recomputing the query from scratch.
\update{
Future research can follow our approach but implement e-graphs directly on top of
 a relational database engine, 
 inheriting the incrementality from the underlying system. 
For example, a datalog engine provides semi-naive evaluation.}

There is opportunity to improve relational \ematching\ even without a
 full-fledged IVM solution.
For example, we have shown in \autoref{fig:speedup} that index building
 can take up a significant portion of the run time.
Our index implementation is based on a hash trie,
 which is simple but difficult to update efficiently.
We are experimenting with an alternative index design based on sort tries,
 in the hope that it can make updates as simple as inserting into a sorted array.

\subsection{Building on Existing Database Systems}
Given our reduction from \ematching\ to conjunctive query answering,
 one may wonder if other \egraph operations could be reduced to
 relational operations so that a fully functioning \egraph engine
 can be implemented purely on top of an off-the-shelf database system.
There are many benefits to it.
For example, we could enjoy an industrial-strength query optimization
 and execution engine for free
 (although most industrial databases do not use worst-case optimal join algorithms),
 and eliminates the cost of transforming an \egraph
 to a relational database.
Moreover, this approach would enjoy any properties of the host database
 system, including persistence, incremental maintenance, concurrency, and fault-tolerance.
% {\color{blue} (It may also be worth mentioning the idea of scalability / persistent e-graphs / egraph as a
%  service / etc.)}

As a proof of concept, we implemented a prototype \egraph implementation on top of SQLite,
 an embedded relational database system, with 160 lines of Racket code.
\Egraph operations like insertion and merging are translated into
 high-level SQL queries and executed
 using SQLite.
% Therefore, the costs of rebuilding both the relational database and its indices
%  are naturally avoided,
%  which can sometimes be the bottleneck.
This na\"ive prototype is not competitive with highly optimized
 implementations like \egg,
 especially given our relational \ematching approach.
However, with appropriate indices and query plan,
 it could have similar asymptotic performance.
 % the standard \egraph algorithm.
Specialized data structures to represent equivalence relations~\cite{nappa2019fast}
 could also help performance.
Therefore, not only \ematching but also other \egraph operations can be
 expressed as relational queries,
 which hints at the possibility of developing real-world \egraph engines
 on top of existing relational database systems.

\section{Conclusion}
\label{sec:conclusion}

In this paper, we present relational \ematching, a novel \ematching~algorithm
 that is conceptually simpler, asymptotically faster, and worst-case optimal.
We reduce \ematching\ to conjunctive queries answering,
 a well-studied problem in databases research.
This relational presentation provides a unified way
 to exploit in query planning not only structural constraints,
 but also equality constraints,
 which are constraints that traditional \ematching\ algorithms cannot effectively leverage.
We implement relational \ematching\ with the worst-case optimal generic join algorithm,
 using which we derive the first data complexity for \ematching.
We integrate our implementation in the state-of-the-art equality saturation engine
\egg, and show relational \ematching\ to be flexible (readily supports multi-patterns)
and efficient (achieves orders of magnitude speedup).

\begin{acks}
This work was supported by the Applications Driving Architectures (ADA)
Research Center, a JUMP Center co-sponsored by SRC and DARPA. This material is based upon work supported by the National Science Foundation
under Grant No. 1749570. Any opinions, findings, and conclusions or recommendations expressed in this
material are those of the author(s) and do not necessarily reflect the views of
the National Science Foundation.
\end{acks}

\bibliographystyle{ACM-Reference-Format}
\bibliography{main}

\end{document}